\documentclass{llncs}

\pdfoutput=1

\usepackage[utf8]{inputenc}
\usepackage[T1]{fontenc}
\usepackage{fullpage}
\usepackage{amsmath}
\usepackage{amssymb}
 \usepackage{amsthm}
 \usepackage{amsfonts}
\usepackage{physics}
\usepackage{dsfont}
\usepackage{mathtools}
\usepackage{url}
\usepackage{hyperref}
\hypersetup{colorlinks=true}
\usepackage{algorithm}
\usepackage{algorithmic}
\usepackage[table,xcdraw]{xcolor} \usepackage{adjustbox}
\usepackage{multirow}
\usepackage [operators,sets] {cryptocode}
\usepackage{bm}
\usepackage{enumitem}
\setlist{nosep} \setlist[1]{labelindent=\parindent}
\usepackage{qcircuit}

\usepackage{booktabs} 

\makeatletter
\AtBeginDocument{\def\doi#1{\url{https://doi.org/#1}}}
\makeatother

\usepackage{caption}
\usepackage{subcaption}

\floatstyle{ruled}

\newfloat{protocol}{htb!}{idf}
\floatname{protocol}{Protocol}

\newfloat{resource}{htb!}{idf}
\floatname{resource}{Resource}

\newfloat{simulator}{htb!}{idf}
\floatname{simulator}{Simulator}

\newfloat{problem}{htb!}{idf}
\floatname{problem}{Problem}

\newcommand{\cptp}[1]{\mathsf{#1}} \newcommand{\pd}[1]{\mathcal{#1}}  \newcommand{\sch}[1]{\bm{#1}} \newcommand{\grp}[1]{\mathrm{#1}} 

\newcommand{\id}{\mathds{1}}

\newcommand{\Id}{\cptp{I}}
\newcommand{\I}{\cptp{I}}
\newcommand{\X}{\cptp{X}}
\newcommand{\Z}{\cptp{Z}}
\newcommand{\Y}{\cptp{Y}}
\newcommand{\Ha}{\cptp{H}}
\newcommand{\T}{\cptp{T}}

\newcommand{\CZ}{\mathsf{CZ}}
\newcommand{\CNOT}{\mathsf{CNOT}}

\newcommand{\Abort}{\mathsf{Abort}}
\newcommand{\BQP}{\mathsf{BQP}}

\newcommand{\acc}{\mathsf{Acc}}
\newcommand{\rej}{\mathsf{Rej}}

\newcommand{\zwt}{\mathrm{zwt}}

\newcommand{\one}{\mathds{1}}

\newtheorem*{thm}{Theorem}
 
\title{Asymmetric Quantum Secure Multi-Party Computation \newline With Weak Clients Against Dishonest Majority}

\author{
Theodoros Kapourniotis\inst{1} \and
Elham Kashefi\inst{2,3} \and
Dominik Leichtle\inst{3} \and
Luka Music\inst{4} \and
Harold Ollivier\inst{5}
}

\institute{
Department of Physics, University of Warwick, Coventry CV4 7AL, United Kingdom \and
School of Informatics, University of Edinburgh, 10 Crichton Street, Edinburgh EH8 9AB, United Kingdom \and
Laboratoire d’Informatique de Paris 6, CNRS, Sorbonne Université, 4 Place Jussieu, 75005 Paris, France \and
Quandela, 7 Rue Léonard de Vinci, 91300 Massy, France \and
DI-ENS, Ecole Normale Supérieure, PSL University, CNRS, INRIA, 45 rue d'Ulm, 75005 Paris, France
}

\date{}

\begin{document}

\pagestyle{plain}

\maketitle

\begin{abstract}
Secure multi-party computation (SMPC) protocols allow several parties that distrust each other to collectively compute a function on their inputs.
In this paper, we introduce a protocol that lifts classical SMPC to quantum SMPC in a composably and statistically secure way, even for a single honest party.
Unlike previous quantum SMPC protocols, our proposal only requires very limited quantum resources from all but one party; it suffices that the weak parties, i.e.~the \emph{clients}, are able to prepare single-qubit states in the $X-Y$ plane.

The novel quantum SMPC protocol is constructed in a naturally modular way, and relies on a new technique for quantum verification that is of independent interest. This verification technique requires the remote preparation of states only in a single plane of the Bloch sphere.
In the course of proving the security of the new verification protocol, we also uncover a fundamental invariance that is inherent to measurement-based quantum computing.
\keywords{Quantum Verification, Delegated Computation, Secure Multi-Party Computation, Distributed Quantum Computing.}
\end{abstract}

\section{Introduction}

\subsection{Motivation} 

Secure Multi-Party Computation (SMPC) protocols allow several parties who
do not trust one another to collectively compute a function
on their inputs. This question was first considered by
Yao~\cite{Y86:how} and has been developed extensively in various
settings (see~\cite{CDN15:multi} and references therein). Several security
guarantees can be provided by such protocols depending on the 
setting: all parties can be on an equal footing,
doing each their share of the computation, or one can handle the brunt of the
computation while all others provide the data. In the first case,
the security goal is to maximise the privacy of the data, while in the
latter it extends to the privacy of the computation which is delegated to the server.

Practical computationally-secure protocols have been
developed and implemented in commercial solutions for protecting
classical multi-party computations. In the quantum case, several
concrete protocols have been proposed (see \S~\ref{sec:rel_work}). In
the circuit model, the state-of-the-art protocol~\cite{DGJM+20:secure} provides an
information theoretic upgrade of classical SMPC that can withstand a
dishonest majority. In the measurement-based model, where weakly
quantum clients delegate their computation to a powerful server, the
best protocol~\cite{KP17:multiparty} does not provide verification of the computation and
settles instead for blindness (i.e.~privacy) of the data when there is
no client-server collusion.

In this work, we show that this difference is not due to the asymmetry
of the clients-server setting. We introduce for this specific
situation a statistically secure lift of a classical SMPC protocol to
a quantum one that provides blindness and verification for
$\mathsf{BQP}$ computations. It remains secure so long as a single
client is honest, thus withstanding possible collusions between
dishonest clients and the server. Building on the techniques
introduced in~\cite{KKLM+22:framework}, its security is proved in the
Abstract Cryptography (AC) framework. The protocol is highly modular
and can tolerate a fixed amount of global noise during the quantum computation
without aborting nor compromising statistical security. Additionally,
it has no space overhead compared to an unprotected delegated computation, thereby allowing clients to use
the server's full power to perform their desired computation,
while security comes only at the price of a polynomial number of
repetitions.

\subsection{Related Work}\label{sec:rel_work}

Quantum SMPC is a long standing research topic in quantum cryptography,
with several directions being explored in the past two
decades.

The first one started with~\cite{CGS02:secure}.  Along with the
introduction of the concept itself, it provided a concrete protocol
for performing such computations in the quantum circuit model.
It guarantees the security of the computation as long as the
fraction of malicious parties does not exceed \(1/6\).
This work has been later extended in~\cite{BCGH+06:secure}, lowering
the minimum number of honest players required for security to a strict
majority.

The second focuses on the interesting edge case of two-party
quantum computations. Several constructive results have been proposed
in the circuit model. In~\cite{DNS10:secure}, a protocol was
introduced and proven secure for quantum honest-but-curious
adversaries. This restriction on the adversaries was removed in
\cite{DNS12:actively} which proved security in the fully malicious
setting and with negligible security bounds. The measurement-based
model of quantum computation has also been considered for constructing
secure two-party quantum computations as it provides a different set
of tools than the circuit model. Verifiable Blind Quantum Computation (VBQC)
first was introduced in~\cite{FK17:unconditionally} in this model,
followed by optimised protocols~\cite{KW17:optimised,KDK15:optimising}.
In~\cite{KW17:garbled} a protocol was proposed in this setting and
proven secure against honest-but-curious adversaries.
In~\cite{KMW17:quantum} this result was extended to fully malicious
adversaries with inverse-polynomial security using the Quantum
Cut-and-Choose technique. More recently, a round-optimal protocol was given
in~\cite{BCKM21:the} based on Oblivious Transfer and LWE, showing that 
two-party quantum computation tasks can be performed in as little as three 
rounds in the CRS model, and two if quantum pre-processing is allowed.

A third set of results focuses on the composability of such protocols,
as earlier results didn't satisfy this property.  Bit commitment was
shown to be complete in the Quantum Universal Composability framework
of~\cite{U10:universally}, meaning that it is sufficient for
constructing quantum or classical SMPC
if parties have access to quantum channels and operations.
This result was later extended in~\cite{FKSZ+13:feasibility,DFLS16:adaptive},
which gives a full analysis of feasibility and completeness of
cryptographic primitives in a composable setting.

More recently, building on these previous works, new concrete
protocols have been proposed to decrease the restrictions
on adversaries while also providing composable security. In the
circuit model, a composably-secure protocol has been introduced
in~\cite{DGJM+20:secure}. It is an extension of~\cite{DNS12:actively}
that is able to cope with a dishonest majority, but which relies on a
complete graph for quantum communication and on a large number of
quantum communication rounds together with powerful quantum
participants. In the MBQC model,~\cite{KP17:multiparty} describes a
protocol that is composable, can tolerate a dishonest majority and
allows the clients to delegate the quantum computation to a powerful
server. Its security is an information-theoretic upgrade of the
classical SMPC primitive used for constructing the protocol. It is
however limited by the absence of verifiability of outputs and the
impossibility to tolerate client-server collusions. Other protocols
have been proposed in alternative models or with different trust
assumptions such as~\cite{HHTF18:composable,LRW20:secure}.
Finally, recent protocols for secure delegated quantum computations
can be run even by purely classical clients. These have been lifted 
to a multi-client setting in \cite{B21:secure} while at the
same time optimising the number of classical rounds of communication. 
This is however at the cost of a larger computation space on the server's device,
which needs to be able to perform QFHE computations of functions large enough
to be computationally-secure.

A subset of the authors proposed an earlier protocol for
QSMPC~\cite{KKMO21:delegating} which comprised a blind pre-computation
step meant to produce a resource state that could then be used to
perform VBQC. This pre-computation turned out to be vulnerable to an
attack that can be applied blindly by the server while having an
effect only on some specific types of qubits thereby compromising
security of the whole protocol. While the present work is a complete
redesign of the protocol that shows improved performance, we include
in~\S~\ref{sec:post-mortem} an in-depth analysis of the shortcommings
of the previous construction. This might be a useful tool to revisit
earlier work where a similar blind pre-computation step is used.

\subsection{Overview of the Protocol and Results}\label{sec:outline}

In this paper, we consider the setting where several weakly quantum
clients want to securely delegate their quantum computation to a
powerful server. The proposed construction turns a single-client MBQC-based
protocol into a multi-party one. More precisely, we use single-client
Secure Delegated Quantum Computing (SDQC) protocols obtained through
the techniques presented in~\cite{KKLM+22:framework}. Such protocols
interleave several computation rounds and test rounds, the latter of
which correspond to stabiliser measurements of the MBQC resource
graph-state used to perform the computations. In such a protocol,
the client must perform two different tasks. First, it has to
prepare encrypted single-qubit states and send them to
the server. This prevents the server from distinguishing computation and
test rounds and also hides the client's data. Then, the client uses the
classical encryption key as well as the measurement outcomes reported by
the server to classically drive the computations and tests performed by the server
on these encrypted qubits. Hence, turning this protocol into a
multi-party one amounts to finding (i) an appropriate single-client SDQC
protocol that will (ii) be composed with a secure collaborative remote
state preparation for the single qubit encrypted states and that will
(iii) be driven collaboratively to perform and verify the desired
computation.

In \S~\ref{sec:single_plane}, we describe a single-client SDQC
Protocol using only $\ket{+_\theta} = (\ket 0 + e^{i\theta}\ket
1)/\sqrt 2$ states, based on the generic single-client SDQC Protocol
of~\cite{KKLM+22:framework}.
This was an open question in the field as all previous SDQC protocols in the MBQC framework with a formal security analysis
use computational basis states (called dummies) to
isolate single qubits in the computation graph. These remain unchanged 
if the server is honest and can be used as traps to detect 
deviations. To overcome this restriction, we must ideally find
a generating set of stabilisers of the graph state for the client's
computation that can be
written with $\Id$, $\X$ and $\Y$ Paulis only.

However, while it is possible to construct
$N-1$ independent stabilisers of this form -- where $N$ denotes the number of
vertices of the graph -- it seems that the stabiliser 
which consists of $\Z$ operators on odd-degree vertices of the graph cannot be generated.
This therefore corresponds to a server's deviation which cannot be caught by our tests
on graphs containing odd degree nodes. If this attack would corrupt the client's computation,
the whole protocol would be insecure.
Fortunately, this is not the case for classical input/output
computations. Indeed, we prove that this deviation corresponds to a
server which has chosen a different orientation of the $Z$ axis compared to the
client. Because inputs are prepared in the $X-Y$ plane and
outputs are projected onto it, we show that this has no effect on
the outcome of the computation.
As a consequence, it is not necessary to detect this specific
deviation by the server to verify the computation. This proves that
the generic single-client SDQC Protocol of~\cite{KKLM+22:framework} can be used to produce secure
dummyless protocols.\footnote{Note that here has been a previous
protocol for dummyless verification~\cite{FKD18:reducing}, whose
security analysis didn't take into account the above deviation. Our
proof of invariance of MBQC to this specific error shows that this
deviation does not constitute a security threat to the protocol in~\cite{FKD18:reducing}.}
\begin{thm}[Informal]
For any graph $G$, there exists a single-client statistically secure
SDQC protocol in the Abstract Cryptography framework that
requires the client to only prepare states in the $X-Y$ plane.
\end{thm}

We then focus on turning this new single-client protocol into a
multi-party one. In \S~\ref{subsec:coll-st-prep}, we
introduce a Collaborative Remote State Preparation (CRSP) protocol. We
show that this gadget (Protocol~\ref{proto:col_state_prep}) securely
implements Remote State Preparation (Resource~\ref{res:rsp}), which allows
a classical party request any $\ket{+_\theta}$ state to be
prepared on the server's device with the help of clients preparing
single qubit states in the $X-Y$ plane.
\begin{thm}[Informal]
  The CRSP gadget is a statistically secure implementation of the Remote State Preparation Resource in the Abstract Cryptography framework.
\end{thm}
The second set of tasks in the single-client protocol, i.e.~choosing the measurement angles of the various
computation and test rounds according to the states prepared using CRSP,
only involve classical computations.
These can be performed using a composably secure classical SMPC.\footnote{The Abstract Cryptography framework used in this work is equivalent to the Quantum Universal Composability (Q-UC) Model of~\cite{U10:universally} if a single Adversary controls all corrupted parties -- which is the case here. Therefore any Classical SMPC protocol which is secure in the Q-UC model can be used to instantiate this functionality. }

In \S~\ref{sec:qsmpc}, we compose the CRSP gadget, classical SMPC, and the
dummyless SDQC protocol into a complete quantum SMPC protocol
(Protocol~\ref{proto:mpqc}). Its outline is:
\begin{enumerate}
\item The clients use the CRSP gadget to prepare $\ket{+_\theta}$
  states on the server's side.
\item They use the classical SMPC together to drive and verify the
  single-client SDQC protocol.
\item Upon acceptance, the results and decryption keys are sent by the
  classical SMPC to each client.
\end{enumerate}
The security proof relies on the composable security of all three
ingredients. Because the CRSP gadget and the dummyless protocol are
statistically secure, this is a direct upgrade of classical to quantum SMPC.
\begin{thm}[Informal]
Composable classical SMPC can be lifted to perform robust quantum SMPC for
$\mathsf{BQP}$ computations in a statistically secure way,
such that all parties but one are restricted to singe-qubit preparations.
\end{thm}
We note that this protocol requires no additional resources in terms of hardware or
quantum communication from the client's side compared to the single-client protocol.
The server only needs to be able to perform the CRSP gadget in addition to the 
operations required by the single-client protocol.

\subsection{Discussion}

In the course of constructing our protocol, we have built two new
ingredients that we believe are of independent interest.

The first one is the Collaborative Remote State Preparation
gadget. Its main feature is to provide some privacy amplification for
the classical-quantum correlations that clients share with the
server. Interestingly, we give evidence that it is hard to
construct a generic gadget that would have similar features for
correlations outside of a single plane of the Bloch sphere, while
retaining its usefulness for cryptographic purposes.
We leave it as an open question to prove a full no-go theorem
in the Abstract Cryptography framework to further
explore what seems to be a deep difference between classical and
quantum input-output computations. Note also that this work supersedes
a previous effort to construct a quantum SMPC protocol in the
clients-server setting with quantum input and outputs. The proposed
construction was similar in spirit with a collaborative remote state
preparation gadget that allowed to prepare encrypted $X-Y$ plane
states but also dummies. However, we give an attack on multiple approaches which
were explored to perform this task, further strengthening the belief that
such cryptographic protocols are hard if not impossible to construct.

The second new ingredient of our proof is the first dummyless SDQC protocol. 
Outside of the specific
purpose of quantum SMPC, it exemplifies the usefulness of the general
tests that were introduced in~\cite{KKLM+22:framework}. By reducing
the requirements on the client side, it also possibly decreases a
source of errors in physical implementations as it is not uncommon
that rotations around one specific axis of the Bloch sphere are
notably easier to perform than others. We also strongly believe that
similar approaches, where traps are tailored to specific settings,
will find applications in the future. Additionally, we show that while
dummyless tests were not enough to detect all deviations, it is
possible to nonetheless verify computations thanks to an as of now unknown
invariance in MBQC. This raises the question of whether it is possible to do
this on purpose, and engineer an invariance in order to lighten the
constraints on the error-detection scheme that the traps implement.

Finally, note that because all SDQC protocols constructed from the
generic protocol of~\cite{KKLM+22:framework} are robust to a fixed amount of
global noise, so is our new multi-party protocol. While not being enough
to scale to large quantum computations, it opens the possibility to
implement experimental proof-of-concepts without resorting to error
correction on near term devices.

\subsection{Organisation of the Paper}

In \S~\ref{sec:single_plane} we construct a single-client SDQC
Protocol using only $\ket{+_\theta} = (\ket 0 + e^{i\theta}\ket
1)/\sqrt 2$
states. \S\S~\ref{subsec:framework}-\ref{sec:dummyless_verification}
construct a family of such schemes and prove their security, while
\S~\ref{subsec:concrete-tests} provides optimised protocols for
various classes of MBQC resource graph-states. In
\S~\ref{subsec:coll-st-prep}, we introduce a Collaborative Remote
State Preparation (CRSP) protocol and prove its security in the AC
framework. In \S~\ref{sec:qsmpc}, we compose the CRSP Protocol, the
dummyless SDQC Protocol and a classical SMPC into a complete quantum
SMPC Protocol (Protocol~\ref{proto:mpqc}) for $\mathsf{BQP}$
computations. In \S~\ref{sec:discussion}, we provide an in-depth
comparison with other protocols, give arguments justifying the
proposed construction -- especially the need for a dummyless SDQC
Protocol -- and discuss some open questions.

Some preliminary notation and material can be found in the
corresponding sections of the Auxiliary Supporting Material part:
Abstract Cryptography in \S~\ref{subsec:ac}, MBQC computation in
\S~\ref{subsec:mbqc} and the Universal Blind Quantum Computation
(UBQC) Protocol in \S~\ref{subsec:ubqc}. A detailed analysis of a
previous attempt at constructing quantum SMPC for weakly quantum
clients is provided in~\S~\ref{sec:post-mortem}.

\section{Verification with States in a Single Plane}
\label{sec:single_plane}
\subsection{A Framework for Verification}
\label{subsec:framework}
The goal of the protocol presented in this section is to construct the
Secure Delegated Quantum Computation Resource~\ref{res:sdqc} (SDQC),
introduced by~\cite{DFPR14:composable}. It allows a single Client to
run a quantum computation on a Server so that the Server cannot
corrupt the computation and doesn't learn anything besides a
controlled leakage $l_{\rho}$ about the Client's computation and
input. The value of $l_{\rho}$, as a function of inputs and
computation, is specified by each protocol.
\begin{resource}[ht]
  \caption{Secure Delegated Quantum Computation with Classical Inputs
    and Outputs}
  \label{res:sdqc}
  \begin{algorithmic}[0]
    \STATE \textbf{Inputs:} 
    \begin{itemize}
    \item The Client inputs a bit-string $x$ and the classical
      description of a unitary $\cptp U$.
    \item The Server chooses whether or not to deviate. This interface
      is filtered by two control bits $(e, c)$.
    \end{itemize}
    
    \STATE \textbf{Computation by the Resource:}
    \begin{enumerate}
    \item If $e = 1$, the Resource sends the leakage $l_{\rho}$ to the
      Server's interface and awaits further input from the Server; if it
      receives $c = 1$, the Resource outputs $\Abort$ at the Client's
      output interface.
    \item If $c=0$, it outputs $O = \mathcal{M}_C \circ \cptp U \ket{x}$
      at the Client's output interface, where $\mathcal{M}_C$ is a
      computational basis measurement.
    \end{enumerate}
  \end{algorithmic}
\end{resource}

Several protocols implementing this resource have been constructed in
the past~\cite{GKK19:verification}. Yet, none has the ability to
provide negligible statistical security while having a client sending
\emph{states in a single plane}. To achieve this, we use the
framework from~\cite{KKLM+22:framework} which neatly separates the
various ingredients required to implement SDQC. We start by briefly
summarising the ingredients which are relevant to the present paper.

\paragraph{Reduction to Pauli Deviations.}
Using the UBQC Protocol~\ref{prot:UBQC} (see appendix) to delegate computations from
Client to Server hides the operations which the Client wishes to
delegate.
The encoding scheme of UBQC naturally imposes a Pauli twirl on any deviation and hence any
attacks by the Server can always be
decomposed as a convex combination of Pauli operators acting on the
qubits of the graph just before performing the $\X$-basis
measurement. Because $\X$ Pauli operators applied in this fashion have
no effect on the computation, as they are absorbed by the $\X$-basis
measurement, we can focus on convex combinations of deviations of the
form $\bigotimes_{v \in V} \Z(v)^{e(v)}$ where the values of $e(v)$
are chosen by the Server and $\Z(v)$ applies the Pauli $\Z$ to qubit
$v$. Such deviation are equivalent to flipping
the measurement outcome for vertices where $e(v) = 1$.

\paragraph{General Strategy for Robust Verification.}
Once all operations delegated to the Server are blind a general
strategy for robust and secure computation follows from the intuition
that (i) correctness is obtained by accepting with overwhelming
probability in the absence of deviation, (ii) security derives from the
ability of the protocol to detect with overwhelming probability all
deviations that potentially affect the computation, and (iii)
robustness follows from accepting additional deviations which have,
with overwhelming probability, no effect on the computation.

\paragraph{Generic Trappified Schemes for Classical I/O.}
With this strategy in mind, a whole class of protocols for verifying
$\mathsf{BQP}$ computations can be easily described. Their flexible design is
able to accommodate objectives that go beyond security, e.g.~for
instance the absence of dummy qubits. These protocols work by
performing separate rounds which are indistinguishable from the Server's point of view, 
some implementing tests, and others computing
$\cptp C$, the Client's target computation. More precisely,
$s$ test rounds and $d$ computation rounds are delegated to the Server
using the UBQC Protocol~\ref{prot:UBQC}, with the requirement that
they share the same graph $G$ and the same order $\preceq_G$ for measuring
the qubits.

Each test round is sampled uniformly at random from a set $\sch P$ of possible
traps called a \emph{trappified scheme}. They each
consist of an input state $\sigma$ which is a tensor product of
single-qubit states, one for each vertex in the graph $G$, a
measurement pattern $T$, and a binary decision function $\tau$. The
test round is accepted when the decision function outputs
$0$ when evaluated on the measurement results returned by the
Server for this trap. It is rejected when the output is $1$. The $d$ computation rounds
correspond to repeating $d$ times the target computation $\cptp C$ on
the target input chosen by the Client using the graph $G$. The outputs
of these computations are then combined through a majority vote. When
all rounds have been executed, the Client accepts if less than a
fixed fraction of test rounds reject. In this case, the output of the
protocol is the result of the majority vote.
The formal protocol is described in Protocol~\ref{proto:dev_detect}.

\begin{protocol}[ht]
  \caption{Trappified Delegated Blind Computation}
  \label{proto:dev_detect}
  \begin{algorithmic}[0]
    
    \STATE \textbf{Public Information:} 
    \begin{itemize}
    \item $G = (V, E, I, O)$, a graph with input and output vertices $I$ and $O$ respectively;
    \item $\sch P$, a trappified scheme on graph $G$;
    \item $\preceq_G$, a partial order on the set $V$ of vertices;
    \item $N, d, w$, parameters representing the number of runs, the
      number of computation runs, and the number of tolerated failed
      tests.
    \end{itemize}
    
    \STATE \textbf{Client's Inputs:} A set of angles $\{\phi_i\}_{i
      \in V}$ and a flow $f$ which induces an ordering compatible with
    $\preceq_G$.
    
    \STATE \textbf{Protocol:}
    \begin{enumerate}
    \item The Client samples uniformly at random a subset $C
      \subset [N]$ of size $d$ representing the runs which will be
      its desired computation, henceforth called computation runs.
    \item For $k \in [N]$, the Client and Server perform the following:
      \begin{enumerate}
      \item If $k \in C$, the Client sets the computation for the run
        to its desired computation $(\{\phi_i\}_{i \in V},
        f)$. Otherwise, the Client samples a test $(T, \sigma, \tau)$
        from the trappified scheme $\sch P$.
      \item The Client and Server blindly execute the run using the
        UBQC Protocol~\ref{prot:UBQC}.
      \item If it is a test, it uses $\tau$ on the measurement
        results to decide whether the test passed or not.
      \end{enumerate}
    \item At the end of all runs, let $x$ be the number of failed
      tests. If $x \geq w$, the Client rejects and outputs $(\bot,
      \rej)$.
    \item Otherwise, the Client accepts the computation. It performs a
      majority vote on the output results of the computation runs and
      sets the result as its output.
    \end{enumerate}  
  \end{algorithmic}
\end{protocol}

\paragraph{Security Conditions for Trappified Schemes with Classical I/O.}
The analysis of the security and robustness properties in the Abstract
Cryptography framework for the resulting protocol depends on two
sets of Pauli operators defined relatively to $\sch P$: the set of
detectable deviations and the set of deviations to which $\sch P$ is
insensitive. These rely on the following definitions, where we use $\pd T$
to denote the probability of the measurement outcomes for a trap $T$
in $\sch P$ and $\cptp E\circ \pd T$ to denote the probability
distribution of measurement outcomes when the deviation $\cptp E$ is
applied to $T$.

\begin{definition}[Pauli Insensitivity]\label{def:insens_tp}
  We say that the trappified scheme $\sch P$ is
  \emph{$\delta$-insensitive to} $\mathcal{E} \subset \mathcal G_V$ if:
  \begin{align}
    \forall \cptp E \in \mathcal{E}, \ \sum_{T\in \sch P}
    \Pr_{\substack{T\sim \pd P \\ t\sim \cptp E\circ \pd T}}[\tau(t) =
      0, T] \geq 1-\delta.
  \end{align}
\end{definition}

\begin{definition}[Pauli Detection]\label{def:detect_tp}
  We say that a trappified scheme $\sch P$ \emph{$\epsilon$-detects}
  $\mathcal{E} \subset \mathcal G_V$ if:
  \begin{align}
    \forall \cptp E \in \mathcal{E}, \ \sum_{T\in \sch P}
    \Pr_{\substack{T \sim \pd P \\ t\sim \cptp E\circ \pd T}}[\tau(t)
      = 1, T] \geq 1-\epsilon.
  \end{align}
\end{definition}

\begin{definition}[Pauli Correctness (Informal)]\label{def:correct_tp}
We say that a computation is correct on deviation $\cptp E$ if the output distribution is the same whether the deviation is applied or not.
\end{definition}

The virtue of defining these properties is that the sets of deviations above can
be characterised efficiently and yield correctness and
security with negligible errors for the overall protocol:
\begin{theorem}[Security of Protocol \ref{proto:dev_detect}, Combining Theorems 8 and 13 from \cite{KKLM+22:framework}]
\label{thm:verif}
  Let $\mathfrak{C}$ be a set of classical $\BQP$ computations on graph $G$.
  Let $\sch P$ be a trappified scheme on graph $G$ that $\epsilon$-detects a set
  of Pauli deviations $\mathcal E_1$ and is $\delta$-insensitive to $\mathcal E_2$ and
  perfectly insensitive to $\Id$.
  Assume that all computations in a set $\mathfrak{C}$ are correct on $\mathcal G_V \setminus
  \mathcal{E}_1$.
  Let $n = s + d$ for $d$ and $s$ proportional to $n$, and $c$ the bounded error of $\BQP$. 
  Let $w$ be the maximum number of test rounds allowed to fail, chosen such that $w < \frac{2c-1}{2c-2}s(1 - \epsilon)$.

  Then Protocol~\ref{proto:dev_detect} $\eta(n)$-constructs the Secure Delegated Quantum Computation
  Resource~\ref{res:sdqc} for computations in set $\mathfrak{C}$ in the Abstract Cryptography framework, where the leak
  is defined as $l_{\rho} = (\mathfrak{C}, G, \sch P, \preceq_G)$, for $\eta(n)$ negligible in $n$.
\end{theorem}

  Note that the value of $\eta(n)$ heavily depends on the value of $\delta$ and $\epsilon$, in particular via the coefficient in the exponential. This means that it is crucial to minimise these detection and insensitivity errors.
  
  Notice also that $w$ is also reliant on $\epsilon$, and minimising this error also allows the protocol to tolerate more honest errors before aborting. 
  This noise-robustness of Protocol~\ref{proto:dev_detect} can be characterised as follows.
  
  \begin{theorem}[Noise-Robustness of Protocol \ref{proto:dev_detect}, Combining Theorems 9 and 13 from \cite{KKLM+22:framework}]
  \label{thm:noise-rob}
  For the same parameter choices as in Theorem \ref{thm:verif}, assume an execution of Protocol \ref{proto:dev_detect}
  with an honest-but-noisy Server such that $p$ is the probability that less than $\frac{w}{s\delta}n$ rounds are affected by a Pauli error.  
  Then the Client accepts the outcome with probability $(1 - p)(1 - \delta')$, for $\delta'$ negligible in $n$.
  \end{theorem}
  
  Since the protocol is secure, we can then guarantee that, if the client accepts, the outcome is also correct up to negligible total variational distance. This means that for machines with a constant amount of global noise below a certain bound, our protocol accepts and yields the correct result with overwhelming probability.

\paragraph{Traps from Stabiliser Tests.}
As a result, the performance of Protocol~\ref{proto:dev_detect} is
governed by the choice of $s$, $d$, $w$ defined above, together with
the error detection and insensitivity capabilities of traps in $\sch
P$.  Ref.~\cite{KKLM+22:framework}~\S~6.1 shows how to construct
general traps from subset stabiliser testing. Indeed, let $\grp S$ be
the stabiliser group for $\ketbra G$ the graph state associated to
$G$, and consider $\{ \cptp S_v = \X(v) \bigotimes_{(v,w)\in E} \Z(w),
v \in V \}$ the set of canonical generators of $\grp S$. One can then
associate a trap to each $\cptp R \in \grp S$ by (i) having the Client
prepare a $+1$ eigenspace of $\cptp R$ as input, and (ii) delegating
to the Server the computation
consisting of measuring $\cptp R$ using the UBQC Protocol~\ref{prot:UBQC}. An accepted trap then corresponds
to the measurement of $\cptp R$ returning the $+1$ eigenvalue.

For the preparation, the client sets each qubit $v\in V$ in the $+1$
eigenstate of $\cptp R(v)$ with $\cptp R(v)$ being uniquely defined by:
\begin{align*}
  \cptp R    & = \bigotimes_{v \in V} \cptp R(v) \\
  \cptp R(v_0) & \in \{\pm 1\} \times \{\I,\X,\Y,\Z\}, \mbox{, for } v_0 = \operatorname{argmin}_{v\in V}{\cptp R(v) \neq \I}  \\
  \cptp R(v) & \in \{\I,\X,\Y,\Z\} \mbox{, for } v \in V \setminus v_0.
\end{align*}
This corresponds to preparing a $+1$ eigenstate of the group generated
by $\{\cptp R(v), v\in V\}$ which contains $\cptp R$ hence satisfying
(i) above.

For the delegated computation consisting of measuring $\cptp R$, the
Client simply instructs the Server to measure each qubit in the
$\X$-basis, getting outcome $t(v)$. The motivation for these
measurements is better understood by examining to which observable
they correspond on the inputs provided by the Client. To this end, one
can conjugate each $\X(v)$ by $\prod_{(v,w)\in E} \CZ_{(v,w)}$, the
entangling operation that the Server performs prior to the
measurement. A simple stabiliser computation shows that $\X(v)$ is
mapped to $\cptp S_v$. That is, measuring $\X(v)$ after the entangling
operation corresponds to measuring $\cptp S_v$ on the inputs provided by the
client. As $\cptp R$ is uniquely defined as $\prod_{v\in\one_{\cptp
    R}} \cptp S_v$ for some set $\one_{\cptp R} \subset V$, and because
$\grp S$ is abelian, the outcome of $\cptp R$ on the input state
provided by the client is the binary sum of the outcomes of
$\cptp S_v$. Using the above correspondence for measurements of $\cptp S_v$ on the
inputs, one concludes that $\bigoplus_{v\in \one_{\cptp R}} t(v)$
determines the outcome of the measurement of $\cptp R$ on the inputs
provided by the Client. Combining the preparation and the measurement,
the Client therefore expects that for an honest Server,
$\bigoplus_{v\in \one_{\cptp R}} t(v) = 1$, thereby fulfilling (ii)
above.

The freedom in choosing which $\cptp R$'s to include in the trappified
scheme $\sch P$ will be at the core of constructing dummyless
verification protocols.

\subsection{A Natural Invariance of MBQC with Classical Input and Output}\label{sec:natural_invariance_mbqc}

In MBQC, computation qubits, i.e.~$v\in O^c$, are measured in the
$\ket{\pm_{\phi'(v)}}$ basis, where $\phi'(v) \in
\Theta=\qty{\frac{k\pi}{4}}_{k \in \qty{0, \ldots, 7}}$ is defined by
the pattern used for the computation. As a result, the computation is
invariant under rotations around the $\phi'(v)$ axis in the $X-Y$
plane just before the measurement. The reason is that such rotations
leave the projectors $\ketbra{+_{\phi'(v)}}$ and
$\ketbra{-_{\phi'(v)}}$ untouched so that it does not affect the
probabilities of the outcomes of a measurement in the
$\ket{\pm_{\phi'(v)}}$ basis. This property is well known and is
actively used in the proof of security of the UBQC protocol as it
allows to fully twirl the deviation of the server on computation
qubits.

If one not only considers local unitary transformations but more
generally local invertible transformations, then MBQC is also
invariant under reflections through the $X-Y$ plane for $v \in
O^c$. The reason is similar to the one given above: such
transformations do not change the projectors onto the
$\ket{\pm_{\phi'(v)}}$ basis and hence do not affect probability
distributions of measurements in the $\ket{\pm_{\phi'(v)}}$ basis.

We will now explore the latter invariance in the special case of
classical input classical output computations where it naturally
extends to the result of the computation itself, as in such case all
qubits are measured in the $X-Y$ plane.

\begin{lemma}\label{lemma:fa_harmless}
	For matrices $\rho = \sum_{\cptp P \in \{\I,\X,\Y,\Z\}^{\otimes n}} \alpha_{\cptp P} \cptp P$ decomposed in the Pauli basis, let $F_A$ be the linear map that applies the reflection through the $X-Y$ plane for all vertices in $A \subset V$, defined as
	\begin{align*}
		F_A(\rho) = \sum_{\cptp P\in \{\I,\X,\Y,\Z\}^{\otimes n}} (-1)^{\zwt_A(\cptp P)} \alpha_{\cptp P} \cptp P,
	\end{align*}
	where $\zwt_A(\cptp P) = |\{ v\in A | \cptp P_v = \Z \}|$ counts the number of vertices in $A$ on which $\cptp P$ equals the Pauli $\Z$. Then, MBQC is invariant under $F_A$ when applied right before the $\ket{\pm_{\phi'(v)}}$ measurements.
\end{lemma}
\begin{proof}
	The probability to obtain the all-zero outcome when measuring all qubits $v \in O^c$ of a state $\rho$ in the $\dyad{+_{\phi'(v)}}$-bases is given by
	\begin{align*}
		\Tr \left( \left( \id_O \otimes \bigotimes_{v\in O^c} \dyad{+_{\phi'(v)}} \right) \rho \right).
	\end{align*}
	Decomposing the above expression in the Pauli basis yields
	\begin{align*}
		&\Tr \left( \left( \sum_{\cptp P' \in \{\I,\X,\Y\}^{\otimes n}} \beta_{\cptp P'} \cptp P' \right) \left( \sum_{\cptp P \in \{\I,\X,\Y,\Z\}^{\otimes n}} \alpha_{\cptp P} \cptp P \right) \right) \\
		&= \sum_{\cptp P' \in \{\I,\X,\Y\}^{\otimes n}} \sum_{\cptp P \in \{\I,\X,\Y,\Z\}^{\otimes n}} \beta_{\cptp P'} \alpha_{\cptp P} \Tr \left( \cptp P' \cptp P \right) = \sum_{\cptp P \in \{\I,\X,\Y\}^{\otimes n}} \beta_{\cptp P} \alpha_{\cptp P} 2^{|V|}.
	\end{align*}
	Calculating the same probability for the all-zero outcome when measuring after applying $F_A$ yields
	\begin{align*}
		&\Tr \left( \left( \id_O \otimes \bigotimes_{v\in O^c} \dyad{+_{\phi'(v)}} \right) F_A (\rho) \right) \\
		&\Tr \left( \left( \sum_{\cptp P' \in \{\I,\X,\Y\}^{\otimes n}} \beta_{\cptp P'} \cptp P' \right) \left( \sum_{\cptp P \in \{\I,\X,\Y,\Z\}^{\otimes n}} (-1)^{\zwt_A(\cptp P)} \alpha_{\cptp P} \cptp P \right) \right) \\
&= \sum_{\cptp P \in \{\I,\X,\Y\}^{\otimes n}} \beta_{\cptp P} \alpha_{\cptp P} 2^{|V|},
	\end{align*}
	and therefore the same value. By an analogous argument, the probabilities for any other outcome coincide as well.
\end{proof}

Note that $F_A (\rho)$ might not always be a physical state.
As a result, if $\ketbra{G}$ denotes the graph state used to implement
classical input classical output MBQC on $G$, one has:
\begin{equation}
  \ketbra{G} = \frac{1}{2^{|V|}}\sum_{\cptp S \in \mathcal{S}} \cptp S, 
\end{equation}
for $\mathcal{S}$ the stabiliser group of the graph state, so that 
for any $\cptp S'\in \grp S$ we have:
\begin{equation}
  \Tr(\cptp S'\ketbra{G}) = \frac{1}{2^{|V|}} \sum_{\cptp S \in \grp S} \tr(\cptp S' \cptp S) = 1.
\end{equation}
In turn, this implies that
\begin{equation}\label{eq:fa_stabilizer}
  \Tr(\cptp S' F_A(\ketbra{G})) = \frac{1}{2^{|V|}} \sum_{\cptp S \in \grp S} (-1)^{\zwt_A({\cptp S})} \Tr(\cptp S' \cptp S) = (-1)^{\zwt_A(\cptp S')}.
\end{equation}
If $F_A(\ketbra{G})$ was a physical state, Equation~\eqref{eq:fa_stabilizer} would imply that it would be stabilised by $(-1)^{\zwt_A(\cptp S)} \cptp S$ for all $\cptp S \in \grp S$. The group structure of stabilisers would then imply that it is also stabilised by the operator $(-1)^{\zwt_A(\cptp S) + \zwt_A(\cptp S')} \cptp S \cptp S'$ for all $\cptp S, \cptp S' \in \mathcal{S}$, and hence $\zwt_A(\cptp S \cptp S') \equiv \zwt_A(\cptp S) + \zwt_A(\cptp S') \pmod{2}$.

However, for $A\subsetneq V$, $\zwt_A(\cdot)$ does not
in general satisfy the above equation. More precisely, take $(v,w) \in
E$, the stabiliser $\cptp S_v \cptp S_w$ will then satisfy $\zwt(\cptp
S_v \cptp S_w) \equiv \zwt(\cptp S_v) + \zwt(\cptp S_w) - 1 \pmod
2$. This is because the overlap of $\cptp S_v$ and $\cptp S_w$ at $v$
will always remove a single $\Z$ coming from $\cptp S_w$, while if the
two stabilisers overlap at some other $\Z$ location in $A$ this will
remove 2 from the weight.

Conversely\footnote{More generally, for disconnected graphs this holds if and only if $A$ is a connected component or a union of connected components.}, setting $A = V$, then indeed $\zwt_V(\cptp S \cptp S') \equiv \zwt_V(\cptp S) + \zwt_V(\cptp S') \pmod{2}$ for all $\cptp S, \cptp S' \in \mathcal{S}$.
Moreover, it is possible to find a unitary
transformation that has the same effect as $F_A$ on $\ketbra{G}$,
implying that $F_A(\ketbra{G})$ is then a physical state,
as witnessed by the following lemma.

\begin{lemma}\label{lemma:fv_as_unitary}
	For any graph $G = (V,E)$ it holds that $F_A(\ketbra{G}) = \cptp U \ketbra{G} \cptp U^\dagger$, where
	\begin{align*}
		\cptp U = \prod_{\substack{v \in V, \\ \deg{v} \equiv 1 \!\!\!\!\! \pmod{2}}} \Z_v
	\end{align*}
	describes the application of $\Z$'s to all odd-degree vertices of $G$.
\end{lemma}
\begin{proof}
	It will be useful to rewrite the stabilisers of $\ketbra{G}$ as follows. For every $\cptp S \in \grp S$, there exists exactly one subset of vertices $V_{\cptp S} \subset V$ such that
	\begin{align*}
		\cptp S = \prod_{v \in V_{\cptp S}} \cptp S_v.
	\end{align*}
	We start with the right side of the equation:
	\begin{align*}
		\cptp U \ketbra{G} \cptp U^\dagger = \cptp U \left( \frac{1}{2^{|V|}}\sum_{\cptp S \in \grp S} \cptp S \right) \cptp U^\dagger = \frac{1}{2^{|V|}}\sum_{\cptp S \in \grp S} \cptp U \left( \prod_{v \in V_{\cptp S}} \cptp S_v \right) \cptp U^\dagger.
	\end{align*}
	Complementing $\cptp U^\dagger \cptp U$ terms, this expression gives:
	\begin{align*}
		\frac{1}{2^{|V|}}\sum_{\cptp S \in \grp S} \prod_{v \in V_{\cptp S}} \cptp U \cptp S_v \cptp U^\dagger.
	\end{align*}
	It is easy to verify that $\cptp U \cptp S_v \cptp U^\dagger = (-1)^{\zwt_V(\cptp S_v)} \cptp S_v$ because of the particular structure of $\cptp U$, and hence the above expression equals
	\begin{align*}
		\frac{1}{2^{|V|}}\sum_{\cptp S \in \mathcal{S}} \prod_{v \in V_{\cptp S}} (-1)^{\zwt_V(\cptp S_v)} \cptp S_v.
	\end{align*}
	Exploiting the additivity of $\zwt_V(\cdot)$, we arrive at
	\begin{align*}
		\frac{1}{2^{|V|}}\sum_{\cptp S \in \mathcal{S}} (-1)^{\sum_{v \in V_{\cptp S}} \zwt_V(\cptp S_v)} \cptp S = \frac{1}{2^{|V|}}\sum_{\cptp S \in \mathcal{S}} (-1)^{\zwt_V(\cptp S)} \cptp S = F_V(\ketbra{G}),
	\end{align*}
	which concludes the proof.
\end{proof}

Combining the statements of Lemma~\ref{lemma:fa_harmless} and Lemma~\ref{lemma:fv_as_unitary}, we finally arrive at the following result, capturing the inherent invariance of classical I/O MBQC to one specific nontrivial error.

\begin{lemma}\label{lemma:harmless}
	Let $G = (V,E)$ be a graph and $\cptp U$ be the unitary operation given by
	\begin{align*}
		\cptp U = \prod_{\substack{v \in V, \\ \deg{v} \equiv 1 \!\!\!\!\! \pmod{2}}} \Z_v,
	\end{align*}
	describing the application of $\Z$'s to all odd-degree vertices of $G$.
	For MBQC on $G$ with classical input and output, the application of $\cptp U$ before the measurements has no effect on the results of the computation.
\end{lemma}

Summarising the results of this section, for any classical-input classical-output MBQC there exists a non-trivial and non-stabiliser deviation that has no influence on the results of the computation. It is important to bear in mind the harmlessness of this error when constructing a verification scheme, as dummyless stabiliser tests will -- by construction -- not be able to detect it.

\subsection{Dummyless Verification}\label{sec:dummyless_verification}

We now arrive at the core of this section: designing single-round
traps restricted to preparing states in the $X-Y$ plane. Using the
construction of traps from Section \ref{subsec:framework}, it amounts to finding a set of
stabilisers of $\ketbra G$ that are only made out of $\I$, $\X$, $\Y$
tensor products.

More precisely, we show that
\begin{lemma} \label{lem:stab_tests}
For any $G = (V,E)$, consider the graph state $\ket G$ and its
stabiliser group $\grp S$. Then, it is always possible to find $|V|-1$
generators of $\grp S$ that are tensor products of $\I$, $\X$ and $\Y$
only.
\end{lemma}

\begin{proof}
  We proceed constructively and exhibit a set of $|V|-1$ generators of
  $\grp R$, subgroup of $\grp S$, and show that $|\grp R| =
  2^{|V|-1}$.

  We start with one such stabiliser, $\cptp R_\text{full} = \prod_v
  \cptp S_v$. This follows simply from
  \begin{equation}
    \cptp R_\text{full}(v) = \X \Z^{\deg(v)},
  \end{equation}
  as for qubit $v$, $\cptp S_v$ contributes to the $\X$ and all
  neighbours contribute a $\Z$ each. Additionally, this shows that for vertices
  $v$ of even degree $\cptp R_{\setminus v} = \prod_{w \in V\setminus v}
  \cptp S_v = \cptp R_\text{full} \cptp S_v$ is also a tensor product
  of $\I$, $\X$, $\Y$. This is because removing $\cptp S_v$ from
  $\cptp R_\text{full}$ leaves an $\I$ at $v$ and changes by one the
  number of $\Z$s on the neighbours of $v$. Unfortunately, removing
  $\cptp S_v$ for $v$ of odd degree leaves a $\Z$ at $v$. To further
  remove this unwanted $\Z$, one can also remove one stabiliser $\cptp
  S_w$ from a neighbouring node $w$ of $v$ from the product. If, in
  addition, $w$ is of odd degree, then the obtained stabiliser will be
  a tensor product of $\I$, $\X$, $\Y$ only. The reason is that at
  $w$, one $\Z$ has been removed when $\cptp S_v$ was removed from
  $\cptp R_\text{full}$ thereby leaving an $\X$ at $w$, so that
  removing $\cptp S_w$ leaves an $\I$. In the general case, one can
  always remove from $\cptp R_\text{full}$ the stabilisers $\cptp S_v$
  along a chain between $u$ and $w$ consisting of even degree nodes
  except for $u$ and $w$ that are odd degree. We denote by $\cptp
  R_{\setminus (u,w)}$ such generator. Note that a given odd-degree
  node will always be in at least one such stabiliser as there are
  always an even number of odd degree nodes in a connected component
  of a graph.

  Now define the group $\grp R$ generated by $\cptp R_\text{full}$,
  $\cptp R_{\setminus v}$ and $\cptp R_{\setminus(u,w)}$ above. Notice
  that multiplying $\cptp R_\text{full}$ with $\cptp R_{\setminus v}$
  gives $\cptp S_v$, so that $S_v$ is in $\grp R$ for even $\deg(v)$.
  Similarly, multiplying $\cptp R_\text{full}$ with $\cptp
  R_{\setminus (u,w)}$ and $\cptp S_v$ for $v$ an even-degree node
  linking $u$ to $w$ shows that any $\cptp S_u \cptp S_w$ with $u$ and
  $w$ odd-degree nodes are also in $\grp R$. Therefore, $\grp R$
  contains all stabilisers that have an arbitrary number of
  even-degree node and an even number of odd-degree ones. Counting the
  number of such stabilisers gives $2^{|V|-1}$ while we know that the
  size of $\grp S$ is $2^{|V|}$, which concludes the proof.
\end{proof}

We now consider the trappified scheme $\sch P$ that can be obtained by
sampling uniformly at random from all these traps rounds. We can
characterise the errors that can be detected by $\sch P$ and those to
which it is insensitive using properties of stabilisers. To this end,
recall that if a Pauli error $\cptp E$ is applied right before the
measurement of a $2$-outcome observable $\cptp M$, then (i) the
measurement outcome probabilities are unchanged if $[\cptp E,\cptp M]
= 0$, and (ii) are swapped for $\{\cptp E, \cptp M\} = 0$. Hence,
whenever $\cptp E$ commutes with $\bigotimes_{v\in\one_{\cptp R}}
\X(v)$ the trap never detects $\cptp E$, whereas it always detects it
whenever it anticommutes. As a consequence, the set of detectable
errors is the set of errors that anticommute with at least one of the
$\bigotimes_{v\in\one_{\cptp R}} \X(v)$ for $\cptp R$ a dummyless trap
measurement.

Hence, for an error $E = \cptp E_Z \cptp E_X$ we need to assess
whether there exists at least one $\cptp R$ in $\sch P$ such that
$|\one_{\cptp R} \cap \one_{\cptp E_Z}| \equiv 1 \pmod 2$ -- where we
have implicitely defined $\cptp E_Z$ (resp.~$\cptp E_X$) as the
operators made of $\Z$s at location of $\Y$ or $\Z$ qubits in $\cptp
E$ (resp.~$\X$ or $\Y$ qubits). To this end, consider $\cptp F$ such
that $\cptp U_G \cptp F = \cptp E \cptp U_G$ where $\cptp U_G$ is the
entangling operation for creating the graph state. Because a trap
amounts to measuring the corresponding stabiliser before the entangling
operation, the above question amounts to knowing whether $\cptp F$
commutes with the stabilisers used to define the dummyless traps of
$\sch P$. Alternatively, we can answer this question by finding out
which Pauli operations commute with all stabilisers defining the
dummyless traps while not being a product of them.

Using Lemma~\ref{lem:stab_tests}, there is one generator $\cptp S_0$
of $\grp S$ that is not in $\grp R$ and such that all errors that
commute with $\grp R$ and are not in $\grp R$ are of the form $\cptp
S_0 \grp R$. From the above description of $\grp R$, $\cptp S_0$ can
be taken as being equal to $\Z$ on all odd-degree nodes. $\cptp S_0$
commutes with all elements of $\grp R$ since they have an even number
of $\cptp S_v$ for $v$ odd-degree, and it is not in $\grp R$ as $\grp
R$ has no element with $\Z$s only. Yet, Lemma~\ref{lemma:harmless}
shows that while $\cptp S_0$ cannot be detected, it is indeed harmless
for the computation.

Hence, we are led to conclude that all possibly harmful errors are
detected by the trappified scheme $\sch P$. Using
\S~\ref{subsec:framework}, we conclude that
\begin{theorem}
Let $G=(V,E)$ be a graph, and $\sch P$ the trappified scheme on $G$
defined by sampling at random from a generating set of $\grp R$ containing only stabilisers with no $\Z$s. Then,
$\sch P$ constructs the SDQC Resource~\ref{res:sdqc} for $\BQP$
computations that can be embedded on the graph $G$ with negligible
correctness and security errors.
\end{theorem}
This follows from the fact that Theorem~\ref{thm:verif} states that a
secure verification scheme can be built from a trappified scheme that
1) detects a specific set $\mathcal{E}$ of $\cptp Z$-Pauli errors, and
2) correctly evaluates the target computation in the presence of any
other $\cptp Z$-Pauli error in $\mathcal{G}_V^{\cptp Z} \setminus
\mathcal{E}$. Lemma~\ref{lemma:harmless} then shows that there is a
specific error $\cptp E^\ast$ which never affects the output
distribution of the target computation and which therefore does not
need to be detected. It hence suffices to find a dummyless trappified
scheme detecting $\mathcal{E} = \mathcal{G}_V^{\cptp Z} \setminus \{
\cptp I, \cptp E^\ast \}$. As shown with Lemma~\ref{lem:stab_tests},
it is indeed possible to find such a trappified scheme. Therefore,
this settles the question whether dummyless verification for $\BQP$ is
possible by the affirmative.

\subsection{Concrete Dummyless Tests}
\label{subsec:concrete-tests}

The previous subsection left open how to concretely construct the
trappified scheme $\sch P$. More precisely, since the efficiency of
the resulting SDQC protocol is tightly linked to the detection
rate of the trappified scheme, it is important to 
minimise its detection, insensitivity and correctness errors. In this section, we discuss the
question of optimising the detection rate. In particular, we construct
concrete dummyless trappified schemes for universal $\BQP$
computations with constant detection rates, independent of the size of
the computation.

\cite{KKLM+22:framework} shows that the general optimisation problem of maximising the detection rate can be expressed in the language of linear programming. Adapted to the case of dummyless trappified schemes, we recall it in the following, as Problem~\ref{prob:optimization_distribution}.
\begin{problem}
\caption{Optimisation of the Distribution of Tests}\label{prob:optimization_distribution}
	\textbf{Given}
	\begin{itemize}
		\item the set of errors $\mathcal{E} = \mathcal{G}_V^{\cptp Z} \setminus \{ \cptp I, \cptp E^\ast \}$ to be detected,
		\item the set of dummyless tests $\mathcal{T}_\text{dummyless}$,
		\item the relation between tests and errors describing whether a test detects an error, $R : \mathcal{T}_\text{dummyless} \times \mathcal{E} \to \bin$,
	\end{itemize}
	\textbf{find} an optimal distribution $p: \mathcal{T}_\text{dummyless} \to [0,1]$ \textbf{maximising} the detection rate $\epsilon \in [0,1]$ \textbf{subject to} the following conditions:
	\begin{itemize}
		\item $p$ describes a probability distribution, i.e. $\sum_{T\in\mathcal{T}_\text{dummyless}} p(T) \leq 1$,
		\item errors are detected at least with the target detection rate, i.e.
			\begin{align*}
				\forall E \in \mathcal{E} : \quad \sum_{\substack{T\in\mathcal{T}_\text{dummyless}\\R(T,E)=1}} p(T) \geq \epsilon.
			\end{align*}
	\end{itemize}
\end{problem}

For any feasible solution to Problem~\ref{prob:optimization_distribution}, the trappified scheme induced by the given distribution of tests gives rise to a secure dummyless SDQC protocol if and only if the detection rate satisfies $\epsilon > 0$.

Recall from Section~\ref{sec:natural_invariance_mbqc} the structure of the harmless error:
\begin{align*}
	\cptp E^\ast = \prod_{\substack{v\in V(G), \\ \operatorname{deg}(v) \equiv 1 \hspace*{-7pt} \pmod{2}}} \cptp Z_v.
\end{align*}
Further, as described in Section~\ref{sec:dummyless_verification}, the set of dummyless tests can be expressed as:
\begin{align*}
	\mathcal{T}_\text{dummyless} = \left\{ \left. \prod_{v \in V_\text{trap}} \X_v \prod_{w \in N_G(v)} \Z_w \; \right| \; V_\text{trap} \subseteq V ,\, \forall v \not\in V_\text{trap} : | N_G(v) \cap V_\text{trap} | \equiv 0 \hspace*{-5pt}\pmod{2} \right\}.
\end{align*}
The last condition ensures that there are no vertices with a single $\Z$ in the respective stabiliser. In this way, every test can be identified with the subset of vertices which act as traps, or equivalently with the complement, the subset of vertices which act as \emph{holes}, i.e. vertices on which the respective stabiliser equals the identity and which can therefore be ignored by the decision function of the trappified scheme. In the following we will also write $V_\text{trap}(T)$ and $V_\text{hole}(T)$ as shorthands for these two sets of vertices.

Analogously, we write $V_\text{error}(\cptp E)$ for the set of vertices on which the error $\cptp E$ is not equal to the identity (and therefore equals the Pauli $\Z$). This makes it easy to give a short description of the relation $R$:
\begin{align*}
	R : (T,\cptp E) \mapsto | V_\text{error}(\cptp E) \cap V_\text{trap}(T) | \hspace*{-5pt} \pmod{2}.
\end{align*}

\paragraph*{Handling Errors on Even-degree Vertices.} As described in Section~\ref{sec:dummyless_verification}, for all even-degree vertices $v\in V$, the test $T$ with $V_\text{hole}(T) = \{ v \}$ is indeed dummyless. Generalising this concept, for any independent set $V^\ast$ of even-degree vertices, we can define a dummyless test $T$ with $V_\text{hole}(T) = V^\ast$. Similarly to the construction of tests in \cite{KKLM+22:framework}, any (fractional) colouring of the vertices of a graph $G$ gives rise to a distribution of independent sets of $G$, and therefore also a distribution of independent sets of even-degree vertices and tests. To this end, let $\mathcal{D}$ be a distribution of independent sets of $G$ such that
\begin{align*}
	\forall v \in V : \Pr_{I \gets \mathcal{D}} \left[ v \in I \right] \geq \frac{1}{\chi_f(G)},
\end{align*}
where $\chi_f(G)$ is the fractional chromatic number of $G$. This distribution exists by definition of the fractional chromatic number. Consider the test strategy given by the distribution $\mathcal{D}_\text{even}$ of tests in $\mathcal{T}_\text{dummyless}$ described as follows:
\begin{enumerate}
	\item Sample an independent set: $V_1 \gets \mathcal{D}$.
	\item Restrict the set to even-degree vertices: $V_2 = V_1 \cap V_\text{even}(G)$, where $V_\text{even}(G) = \{ v \in V \; | \; \deg(v) \equiv 0 \pmod{2} \}$.
	\item Choose a random subset to determine the location of holes: $V_3 \gets \mathcal{U}\left(\wp(V_2)\right)$.
	\item Perform the dummyless test $T$ determined by $V_\text{hole}(T) = V_3$.
\end{enumerate}
As the following Lemma shows, this strategy allows for a detection rate of errors that affect even-degree vertices that scales inversely with the fractional chromatic number of the graph.
\begin{lemma}[Even-degree Error Detection]\label{lemma:even_degree_detection}
	The above-mentioned test strategy $\left(\frac{1}{2\chi_f(G)}\right)$-detects the error set $\mathcal{E}_\text{even} = \{ E \in \mathcal{G}_V^{\cptp Z} \; | \; V_\text{error}(E) \cap V_\text{even} \neq \emptyset \}$, i.e.
	\begin{align*}
		\forall E \in \mathcal{E}_\text{even} : \quad \mathbb{E}_{T \gets \mathcal{D}_\text{even}} \left[ |V_\text{error}(E) \cap V_\text{trap}(T)| \equiv 1 \hspace*{-5pt} \pmod{2} \right] \geq \frac{1}{2\chi_f(G)}.
	\end{align*}
\end{lemma}
\begin{proof}
	Let $E \in \mathcal{E}_\text{even}$. Then, by definition of the test distribution, it holds that
	\begin{align*}
		&\mathbb{E}_{T \gets \mathcal{D}_\text{even}} \left[ |V_\text{error}(E) \cap V_\text{trap}(T)| \equiv 1 \hspace*{-5pt} \pmod{2} \right] \\
		&\geq \mathbb{E}_{V_3 \gets \mathcal{U}(\wp(V_2))} \left[ |V_\text{error}(E) \cap V_3| \equiv 1 \hspace*{-5pt} \pmod{2} \; | \; V_\text{error}(E) \cap V_2 \neq \emptyset \right] \\
		&\qquad\qquad \cdot \Pr_{V_1 \sample \mathcal{D}} \left[ V_\text{error}(E) \cap V_1 \cap V_\text{even}(G) \neq \emptyset \right] \\
		&\geq \frac{1}{2} \cdot \frac{1}{\chi_f(G)},
	\end{align*}
	which concludes the proof.
\end{proof}

\paragraph{Handling Errors on Odd-degree Vertices.} Since all errors acting non-trivially on even-degree vertices are already handled in the previous case, it remains to detect errors that affect only odd-degree vertices and act as the identity on even-degree vertices.

To this end, we construct a specific type of test. For $k\geq 2$, let $(v_1,\dots,v_k) \in V^k$ be a chain of vertices in $G$ satisfying the following conditions:
\begin{enumerate}
	\item The end vertices are of odd degree: $\deg(v_1) \equiv \deg(v_k) \equiv 1 \pmod{2}$.
	\item All intermediate vertices are of even degree: $\deg(v_2) \equiv \dots \equiv \deg(v_{k-1}) \equiv 0 \pmod{2}$.
	\item Only subsequent vertices are neighbours in $G$:
		\begin{align*}
			\forall i,j \in \{1,\dots,k\} : \quad \{v_i,v_j\} \in E(G) \Leftrightarrow |i-j| = 1.
		\end{align*}
\end{enumerate}
It is easy to verify that under these conditions there exists a valid dummyless test $T$ with $V_\text{hole}(T) = \{ v_1, \dots, v_k \}$. Note, that there might not be a chain of this type in $G$ for any pair of odd-degree vertices as end points. However, it is possible to connect any two odd-degree vertices through a chain of chains that might traverse other odd-degree vertices at the end and starting points of chains. In this way, it is possible to choose a ``spanning tree'' of $(|V_\text{odd}(G)|-1)$ chains that connects all odd-degree vertices in the graph $G$.

Define the set of errors on odd-degree nodes only as $\mathcal{E}_\text{odd} = \{ E \in \mathcal{G}_V^{\cptp Z} \; | \; V_\text{error}(\cptp E) \cap V_\text{even}(G) = \emptyset \; \wedge \; V_\text{error}(\cptp E) \cap V_\text{odd}(G) \neq \emptyset \}$ and let $\cptp E \in \mathcal{E}_\text{odd} \setminus \{\cptp E^\ast\}$. Then, there must exist two odd-degree vertices $v_1 \in V_\text{error}(\cptp E)$ and $v_2 \not\in V_\text{error}(\cptp E)$. But then at least one of the chains connecting $v_1$ and $v_2$ with start in a vertex affected by the error $\cptp E$ and end in a vertex unaffected by $\cptp E$. Since all intermediate vertices are of even degree and therefore unaffected by $\cptp E$, the test given by this chain detects $\cptp E$. This essentially shows the following statement.
\begin{lemma}[Odd-degree Error Detection]\label{lemma:odd_degree_detection}
	There exists an efficient testing strategy that $\left( \frac{1}{|V_\text{odd}(G)|-1} \right)$-detects errors in $\mathcal{E}_\text{odd} \setminus \{\cptp E^\ast\}$.
\end{lemma}
Combining the testing strategies from Lemma~\ref{lemma:even_degree_detection} and Lemma~\ref{lemma:odd_degree_detection} immediately yields the following result for testing strategies on general graphs.
\begin{lemma}[Error Detection on General Graphs]
	For any graph $G$, there exists an efficient testing strategy that $\varepsilon$-detects $\mathcal{E} = \mathcal{G}_V^{\cptp Z} \setminus \{\cptp I, \cptp E^\ast\}$, where
	\begin{align*}
		&\varepsilon = \frac{1}{2\chi_f(G)(|V_\text{odd}(G)|-1)} \left( \frac{1}{2\chi_f(G)} + \frac{1}{|V_\text{odd}(G)|-1} \right)^{-1} \\
		&\qquad\qquad \geq \frac{1}{2} \min \left\{ \frac{1}{2\chi_f(G)}, \frac{1}{|V_\text{odd}(G)|-1} \right\}.
	\end{align*}
\end{lemma}
This already shows that the detection rate that is achievable on general graphs decreases at most linearly in the number of vertices of the graph. This lower bound is however far from tight in many cases. In fact, even for universal graph states a constant lower bound is possible as the following result shows.
\begin{figure}[ht]\centering
\begin{subfigure}[t]{0.23\textwidth}
\includegraphics[width=\linewidth]{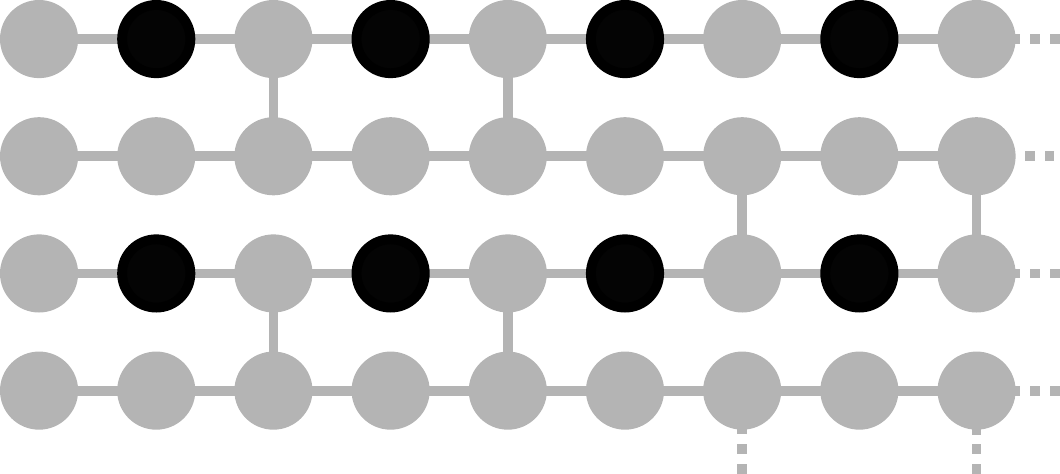}
\caption{Holes in even-degree vertices, type 1.}
\label{fig:brickwork_tests_even_1}
\end{subfigure}
\hspace*{0.15cm}
\begin{subfigure}[t]{0.23\textwidth}
\includegraphics[width=\linewidth]{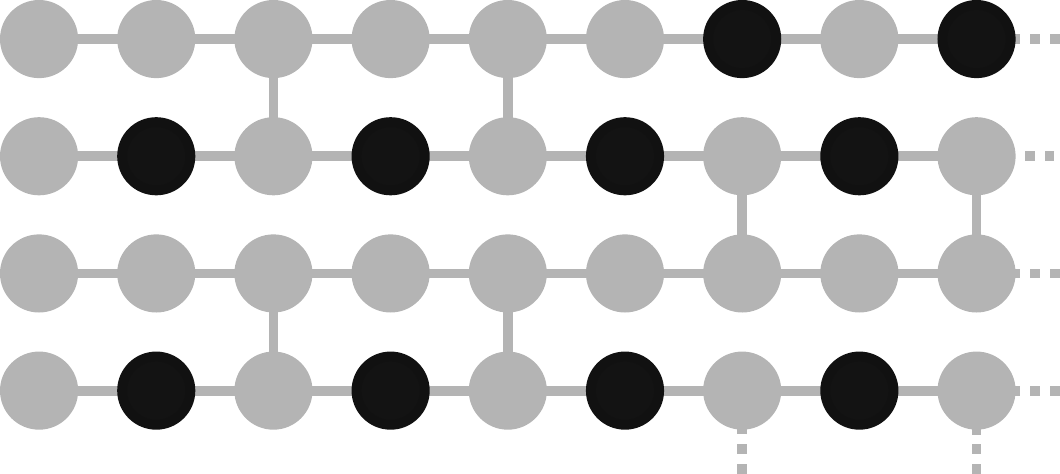}
\caption{Holes in even-degree vertices, type 2.}
\label{fig:brickwork_tests_even_2}
\end{subfigure}
\hspace*{0.15cm}
\begin{subfigure}[t]{0.23\textwidth}
\includegraphics[width=\linewidth]{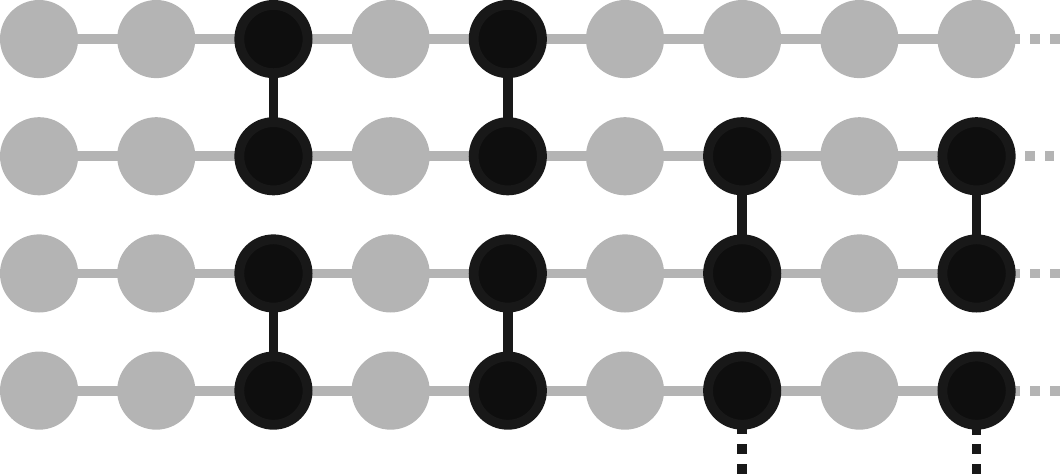}
\caption{Holes in odd-degree vertices, type 1.}
\label{fig:brickwork_tests_odd_1}
\end{subfigure}

\vspace*{0.3cm}
\begin{subfigure}[t]{0.23\textwidth}
\includegraphics[width=\linewidth]{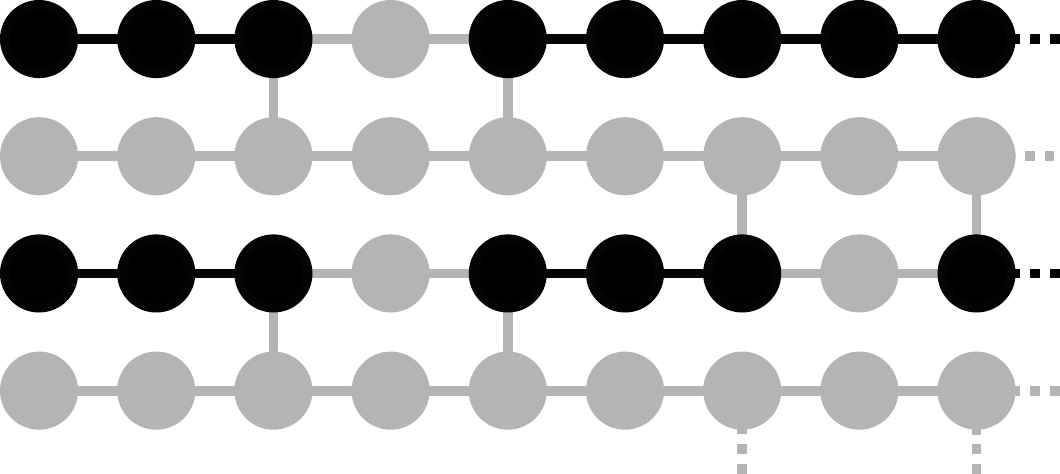}
\caption{Holes in odd-degree vertices, type 2.}
\label{fig:brickwork_tests_odd_2}
\end{subfigure}
\hspace*{0.15cm}
\begin{subfigure}[t]{0.23\textwidth}
\includegraphics[width=\linewidth]{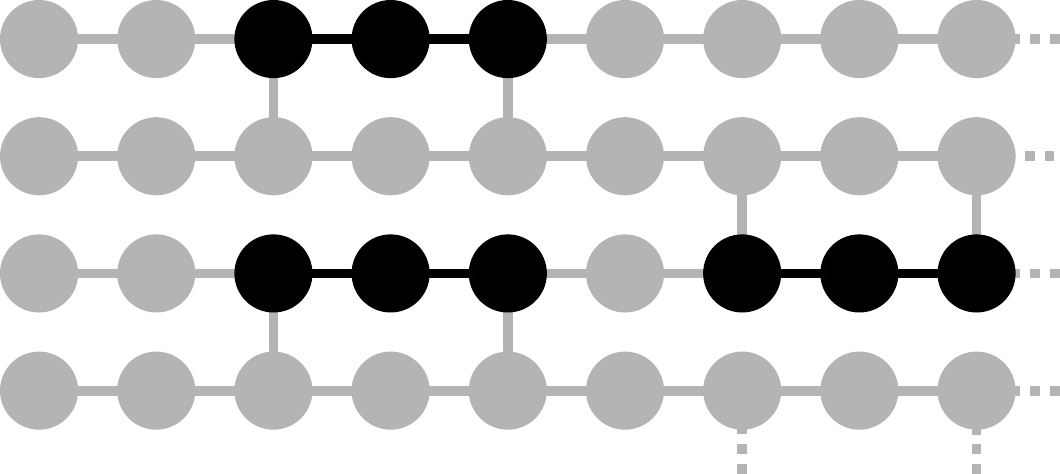}
\caption{Holes in odd-degree vertices, type 3.}
\label{fig:brickwork_tests_odd_3}
\end{subfigure}
\hspace*{0.15cm}
\begin{subfigure}[t]{0.23\textwidth}
\includegraphics[width=\linewidth]{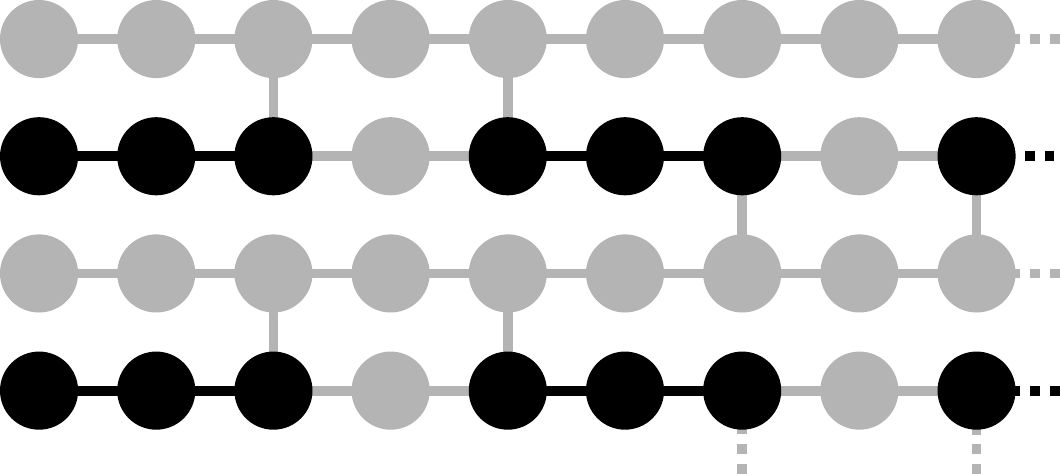}
\caption{Holes in odd-degree vertices, type 4.}
\label{fig:brickwork_tests_odd_4}
\end{subfigure}
\hspace*{0.15cm}
\begin{subfigure}[t]{0.231\textwidth}
\includegraphics[width=\linewidth]{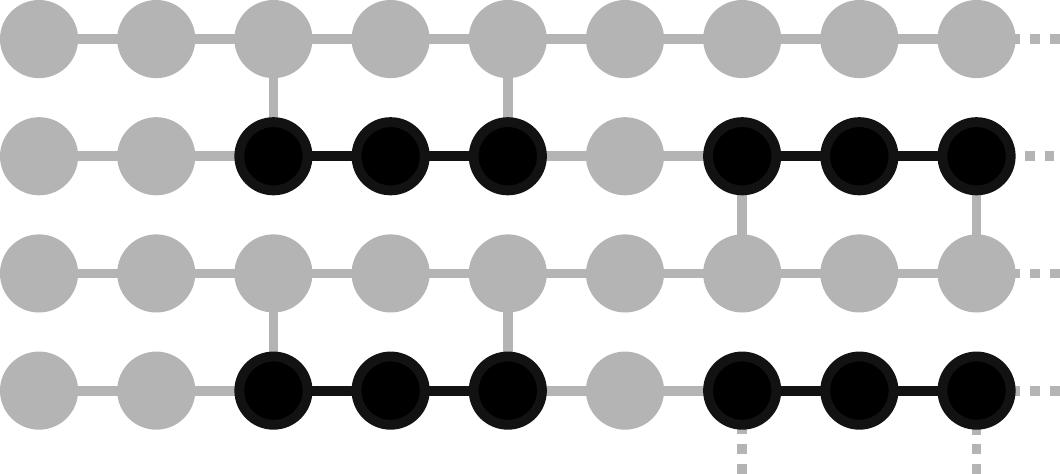}
\caption{Holes in odd-degree vertices, type 5.}
\label{fig:brickwork_tests_odd_5}
\end{subfigure}
\caption{The seven types of dummyless tests for the brickwork graph. A trap configuration is sampled by randomly choosing one of the seven types, and then in cases \ref{fig:brickwork_tests_even_1}-\ref{fig:brickwork_tests_even_2} sampling uniformly at random a subset of marked vertices as holes, and in cases \ref{fig:brickwork_tests_odd_1}-\ref{fig:brickwork_tests_odd_5} sampling uniformly at random a subset of marked chains as holes.}
\label{fig:brickwork_tests}
\end{figure}
\begin{lemma}[Error Detection on the Brickwork State]\label{lemma:brickwork_detection}
	Let $G$ be a brickwork graph. Then, there exists an efficient testing strategy that $\left(1/14\right)$-detects $\mathcal{E} = \mathcal{G}_V^{\cptp Z} \setminus \{\cptp I, \cptp E^\ast\}$.
\end{lemma}
\begin{proof}[Proof Sketch]
To detect errors affecting even-degree vertices, use the strategy from Lemma~\ref{lemma:even_degree_detection}. As the brickwork graph is bipartite, this will yield a detection rate of $1/4$.
	
To detect errors on odd-degree vertices, follow the strategy from Lemma~\ref{lemma:odd_degree_detection}, but use chains that can be tested in parallel to boost the detection rate. There are five classes of chains between odd-degree vertices that can each be run at the same time. One class consists of all vertical chains, and the other four of horizontal chains where every class contains chains only in every second row and only every second horizontal chain on these rows. By testing random subsets of these classes of chains, the detection rate in this case is lower bounded by $1/10$.
	
Optimal switching between these two strategies (with probabilities $2/7$ and $5/7$) yields an overall detection rate of $1/14$. The different types of tests on the brickwork graph are depicted in Figure~\ref{fig:brickwork_tests}.
\end{proof}

\section{Collaborative State Preparation}
\label{subsec:coll-st-prep}

Following the approach outlined in \S~\ref{sec:outline}, we now turn
to the design of a composably secure protocol for implementing the
preparation of the input states required by the dummyless protocols
introduced in~\S~\ref{sec:dummyless_verification}. The Collaborative
Remote State Preparation Protocol \ref{proto:col_state_prep} presented
here will allow $n$ Clients to collaboratively construct an encrypted
state on the Server whose encryption key is held by a purely classical
party called the Orchestrator. It guarantees that no malicious
coalition including up to $n-1$ Clients and the Server (but not the
Orchestrator) has any knowledge about the final state.

This security property is captured formally as follows. The Remote
State Preparation Resource \ref{res:rsp} (or RSP) allows one party
called the Sender to prepare a quantum state on a device held by
another party called the Receiver. Its simplest instantiation requires
only a direct quantum channel between the two participants but more
interesting scenarios can be considered, for example using untrusted
relays or additional participants. We specify this resource for our
specific case, i.e.~sending states in the $X-Y$ plane.

\begin{resource}[ht]
\caption{Remote State Preparation}
\label{res:rsp}

\begin{algorithmic}[0]

\STATE \textbf{Inputs:} The Sender has as input an angle $\theta \in \Theta = \qty{\frac{k\pi}{4}}_{k \in \qty{0, \ldots, 7}}$.

\STATE \textbf{Computation by the Resource:} The Resource prepares and sends the state $\ket{+_\theta}$ to the Receiver.

\end{algorithmic}
\end{resource}

The goal of the Collaborative Remote State Preparation Protocol is
then to construct this Remote State Preparation Resource \ref{res:rsp}
between the Orchestrator and the Server using one Quantum Channel
Resource between each Client and the Server and one Secure Classical
Channel Resource between each Client and the Orchestrator. This latter
Resource transmits faithfully and privately any classical message from
the sender to the receiver, while only leaking the size of the message
to an eavesdropper.

\begin{protocol}[ht]
\caption{Collaborative Remote State Preparation}
\label{proto:col_state_prep}
\begin{algorithmic} [0]
\STATE \textbf{Input:} The Orchestrator has as input an angle $\theta \in \Theta$. The Server and Clients have no input.

\STATE \textbf{Protocol:}
\begin{itemize}
\item Client $j$ samples uniformly at random $\theta_j \in_R \Theta$ and sends $\ket{+_{\theta_j}}$ to the Server.
\item Client $j$ sends $\theta_j$ to the Orchestrator using a Secure Classical Channel.
\item For each $j \neq n$, the Server applies $\CNOT_{n,j}$ between the qubits $n$ and $j$, with the first being the control and the second the target. It measures the target qubit $j$ in the computational basis with measurement outcome $t_j$. It sends the vector $\bm{t}$ containing all the measurement outcomes to the Orchestrator.
\item The Orchestrator computes $\theta' = \theta_n + \sum_{j \in [n-1]} (-1)^{t_j} \theta_j$ and sends a correction $(b, (-1)^b\theta - \theta')$ to the Server, who applies $\X^b\Z((-1)^b\theta - \theta')$ to the unmeasured qubit, keeping it as output.
\end{itemize}

\end{algorithmic}
\end{protocol}

\begin{figure}[ht]\centering
\subfloat[The Server receives the qubits and applies $\CNOT$ gates. The central qubit $n$ is the control, the rest are targets.]{
\label{fig:st-prep-1}
		\includegraphics[width=0.4\textwidth]{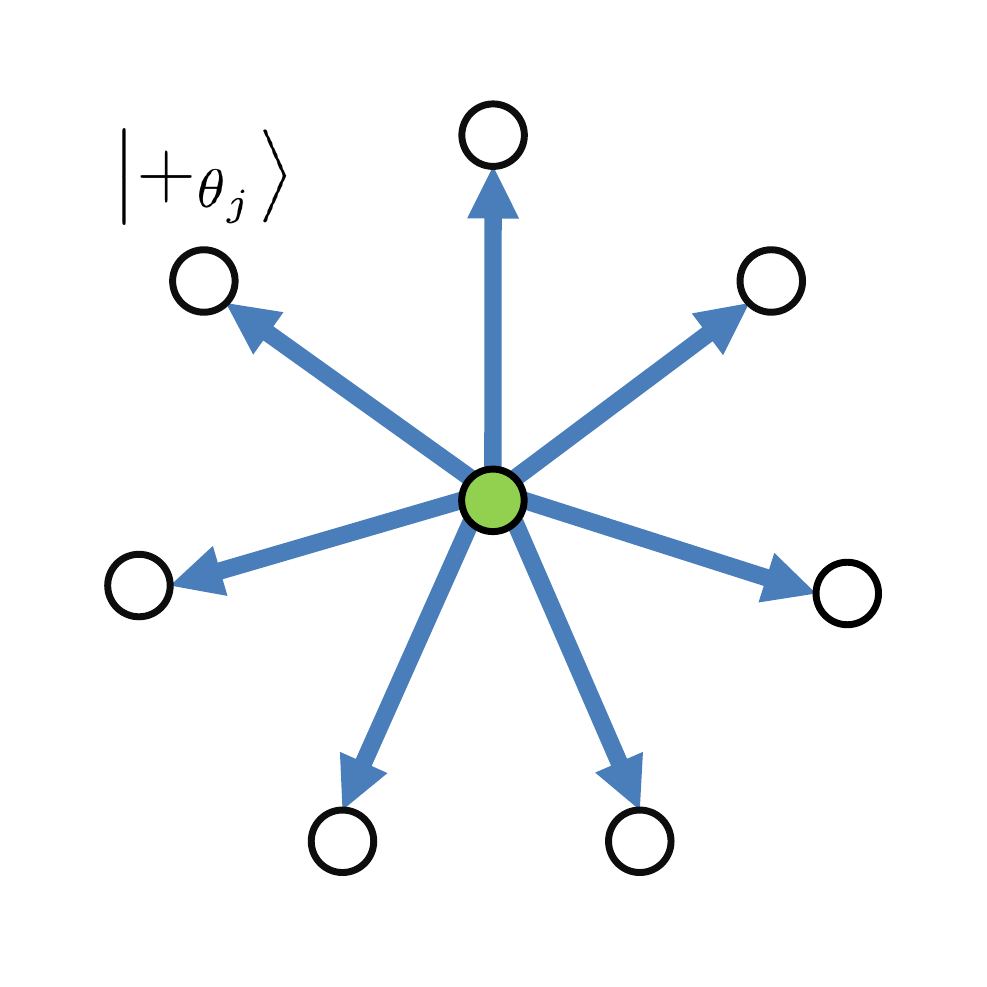}
}
\qquad
\subfloat[The Server measures all qubits but the central one in the computational basis and gets outcomes $t_j \in \bin$.]{
\label{fig:st-prep-2}
		\includegraphics[width=0.4\textwidth]{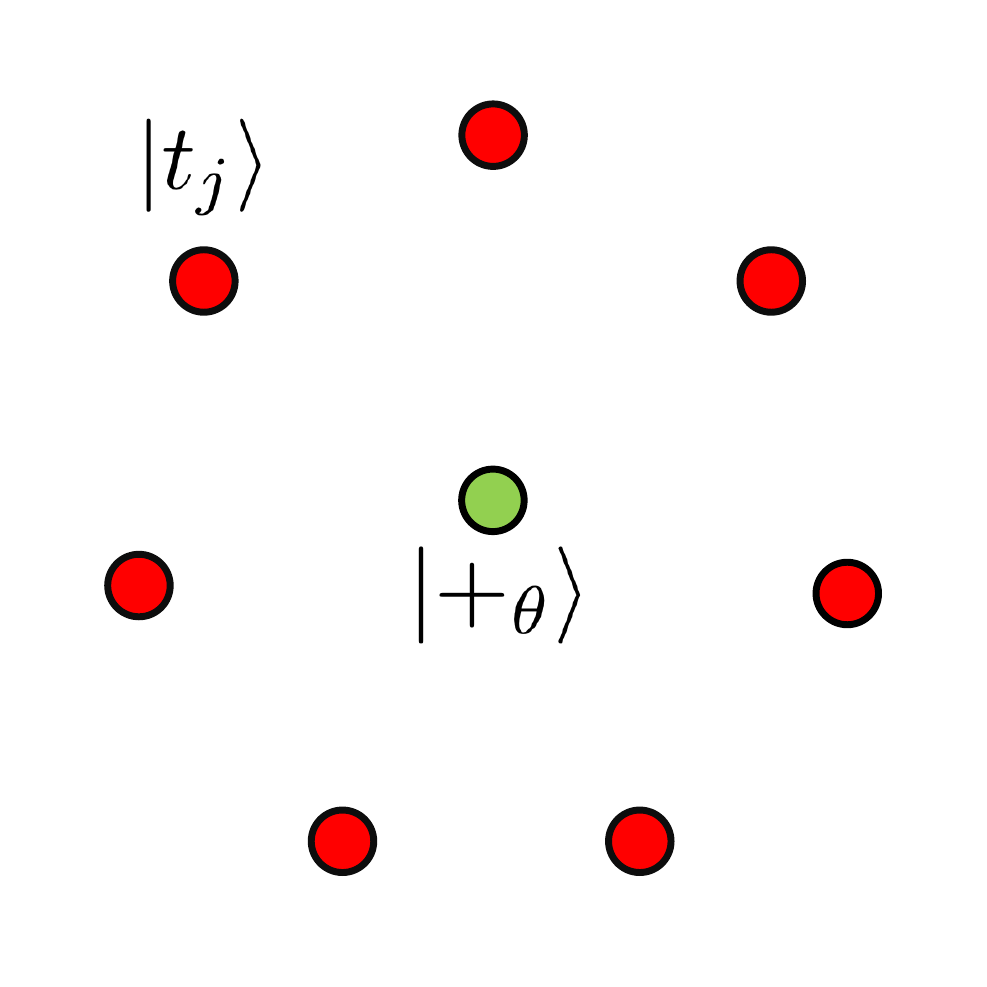}
}
\caption{Collaborative Remote State Preparation for eight qubits. All qubits start in the state $\ket{+_{\theta_j}}$.}
\label{fig:st-prep}
\end{figure}

We can now state the main result of this section, namely the correctness and security of Protocol \ref{proto:col_state_prep} in the AC framework. Both properties are proven independently below.

\begin{theorem}[Security of Collaborative Remote State Preparation]
\label{thm:coll-sec}
Protocol \ref{proto:col_state_prep} perfectly constructs the Remote State Preparation Resource \ref{res:rsp} from Secure Classical Channel Resources between each Client and the Orchestrator, for malicious coalitions that include the Server and at most $n-1$ Clients.
\end{theorem}

\begin{proof}[Proof of Correctness]
The state of the central qubit after an honest execution of Protocol~\ref{proto:col_state_prep} before the correction sent by the Client is $\ket{+_{\theta'}}$ with:
\begin{equation}
\label{eq:theta}
\theta' = \theta_n + \sum\limits_{j \in [n-1]} (-1)^{t_j} \theta_j.
\end{equation}
It is sufficient to prove this for a pure state $\ket{\phi} = \alpha\ket{0} + \beta\ket{1}$ as control. 
We apply a $\CNOT$ gate with $\ket{\phi}$ as control and $\ket{+_{\hat{\theta}}}$ with $\hat{\theta} \in \Theta$ as target, followed by a measurement of this second qubit in the computational basis. 
Let~$t \in \bin$ be the measurement result. After tracing out the second qubit post-measurement, the system is in the following state:

\begin{align*}
\begin{split}
\sqrt{2}\bra{0}_2 \X_2^t \CNOT_{1,2}\ket{\phi}\ket{+_{\hat{\theta}}} =& \bra{0}_2 \X_2^t \qty(\alpha \ket{00} + \alpha e^{i\hat{\theta}}\ket{01} + \beta \ket{11} + \beta e^{i\hat{\theta}} \ket{10}) \\
=& \bra{0}_2(\alpha\ket{0} + \beta e^{i\hat{\theta}} \ket{1})\ket{t} + e^{i\hat{\theta}}\bra{0}_2(\alpha\ket{0} + \beta e^{-i\hat{\theta}} \ket{1})\ket{t \oplus 1} \\
=& \Z(\hat{\theta})\ket{\phi}\braket{0}{t} + e^{i\hat{\theta}}\Z(-\hat{\theta})\ket{\phi}\braket{0}{t \oplus 1}
\end{split}
\end{align*}

\noindent Therefore, the result of this single step is $\Z((-1)^t\hat{\theta})\ket{\phi}$ up to a global phase. Replacing the result above in the sequence of $\CNOT$'s and measurements performed by the Server where the control is qubit $n$ and the targets are qubits $j \neq n$ yields the desired value for $\theta'$. 
Finally, the rotation correction $(-1)^b\theta - \theta'$ sent by the Orchestrator, along with $\X^b$, transform the value of the final state into $\ket{+_\theta}$.
\end{proof}

\begin{proof}[Proof of Security]
We first construct a Simulator against an adversarial Server and a
coalition of $n-1$ Clients, which represents the worst case. The
Server expects to receive $n$ qubits and a final correction after
transmitting the measurement results. The Simulator has single-query
oracle access to the Remote State Preparation Resource \ref{res:rsp}
for state set $\qty{\ket{+_\theta}}_{\theta \in \Theta}$. It receives
a state from this resource, without the corresponding classical
description, and must make the Server accept this state as its output
at the end of the interaction. The actions of this Simulator are
described in Simulator \ref{sim:rot}. Let $h$ be the index associated
to the honest Client.

\begin{simulator}
\caption{Malicious Server and $n-1$ Clients}
\label{sim:rot}
\begin{enumerate}	
	\item The Simulator calls the Remote State Preparation Resource \ref{res:rsp} and receives a state $\ket{+_\theta}$.
	\item It then emulates the behaviour of the $n$ Quantum Channel Resources:
		\begin{itemize}
			\item For indices $j \neq h$, it simply forwards the state from corrupted Client $j$ to the Server;
			\item For index $h$, it samples uniformly at random $\theta_h \in_R \Theta$ and $b_h \in_R \bin$, and sends an encrypted version $\Z(\theta_h)\X^{b_h}(\ket{+_\theta})$ of the state received from the RSP Resource.
		\end{itemize}
	\item It then emulates the Secure Classical Channel Resources and receives from each corrupted Client $j \neq h$ a value $\theta_j$
	\item It receives from the Server a bit-string of measurement results $\mathbf{t} \in [n-1]$.
	\item After extending the bit-string $\mathbf{t}$ with $t_n = 0$, it computes $\theta'$ using Equation \ref{eq:theta} and sends the correction $(t_h \oplus b_h, -\theta')$ to the Server (by impersonating the Orchestrator) and halts.
\end{enumerate}
\end{simulator}

We can now prove that no Distinguisher can tell apart the following
two situations with one honest client: (i) the ideal resource
interacting with the Simulator, and (ii) the real scenario.

\paragraph{Data and transcripts available to the Distinguisher.}
By construction, the Distinguisher fixes $\theta$ the angle of the
desired state to be prepared at the Server output-interface. It also
fixes the value of all $\theta_j$ for $j \neq h$ both in the real and
ideal scenarios and has perfect knowledge of the states sent by
malicious parties. It does not have access to $\theta_h$ as this is
fixed by the honest client protocol.

Before sending the values for the measurement outcomes, the
Distinguisher receives from the non-corrupted party the state
$\ket{+_{\theta_h}}$ in the real case and the state
$\ket{+_{(-1)^{b_h}\theta + \theta_h}}$ in the ideal case. After
sending the bit-string $\mathbf{t}$, regardless of how it was chosen,
the Distinguisher receives a bit and an angle corresponding to the
corrections chosen by either the Orchestrator or the Simulator. In the
first case this is equal to $(b, (-1)^b\theta - \theta')$ and in the
second case $(t_h \oplus b_h, -\theta')$ with $b$ being chosen
uniformly at random and $\theta'$ being computed in the exact same way
in both settings (see Equation~\ref{eq:theta}).

The remaining parameters in the real case and ideal cases are the
received honest state, the associated measurement outcome, the
$\X$-correction bit and the $\Z$-correction angle. This gives us the
following variables that are in the hands of the Distinguisher (rows
are labeled by the meaning of the corresponding data in the real
setting):
\begin{center}
\begin{tabular}{lcc}
\toprule
& \textbf{Real world} & \textbf{Ideal world} \\
\midrule
\textit{Orchestrator-chosen output angle} & $\theta$ & $\theta$ \\ 
\textit{Server's received quantum state} & $\ket{+_{\theta_h}}$ & $\ket{+_{(-1)^{b_h}\theta + \theta_h}}$ \\
\textit{Measurement result bit} & $t_h$ & $t_h$ \\
\textit{Orchestrator correction bit} & $b$ & $b_h \oplus t_h$ \\
\textit{Orchestrator correction angle} & $(-1)^b\theta - (-1)^{t_h}\theta_h$ & $- (-1)^{t_h}\theta_h$ \\
\bottomrule
\end{tabular}
\end{center}

\paragraph{Indistinguishability of data and transcripts for the Distinguisher.}
To finish the security proof, we need to show that the distributions
of the above data and transcripts are statistically indistinguishable
in both scenarios. To do this, we will perform a series of row-wide
operations and eliminate the parameters of the corrupted parties so
that we are left with a new set of variables that will be trivially
indistinguishable. The reversibility of each operation and its
dependency on values that are known to the Distinguisher guarantees
that it can always undo them.

First, multiply the final angle by $(-1)^{t_h}$ and use this angle to
apply a rotation to the state. This transforms the above 
values into:
\begin{center}
\begin{tabular}{cc}
\toprule
\textbf{Real world} & \textbf{Ideal world} \\
\midrule
$\theta$ & $\theta$ \\
$\ket{+_{(-1)^{b \oplus t_h}\theta}}$ & $\ket{+_{(-1)^{b_h}\theta}}$ \\
$t_h$ & $t_h$ \\
$b$ &$b_h \oplus t_h$ \\
$(-1)^{b \oplus t_h}\theta - \theta_h$ & $- \theta_h$ \\
\bottomrule
\end{tabular}
\end{center}
Note that in both cases, the value for $\theta_h$ only appears in the
last row term. Since it is chosen uniformly at random both final terms
follow the same distribution, meaning that they give no distinguishing
advantage. We can therefore safely omit them in the rest of the
process:
\begin{center}
\begin{tabular}{cc}
\toprule
\textbf{Real world} & \textbf{Ideal world} \\
\midrule
$\theta$ & $\theta$ \\
$\ket{+_{(-1)^{b \oplus t_h}\theta}}$ & $\ket{+_{(-1)^{b_h}\theta}}$ \\
$t_h$ & $t_h$ \\
$b$ & $b_h \oplus t_h$ \\
\bottomrule
\end{tabular}
\end{center}
Since $b$ in the first row is a bit sampled uniformly at
random, we can substitute it with $b\oplus t_h$ without changing the
distribution.\footnote{This is the hidden reason for the additional
encryption via $\X^b$ in the protocol.} We arrive at
\begin{center}
\begin{tabular}{cc}
\toprule
\textbf{Real world} & \textbf{Ideal world} \\
\midrule
$\theta$ & $\theta$ \\
$\ket{+_{(-1)^{b}\theta}}$ & $\ket{+_{(-1)^{b_h}\theta}}$ \\
$t_h$ & $t_h$ \\
$b \oplus t_h$ & $b_h \oplus t_h$ \\
\bottomrule
\end{tabular}
\end{center}
Because the $b$ and $b_h$ are uniformly random bits, the above two
distributions are identical, which concludes the proof.
\end{proof}

\section{Quantum Secure Multi-Party Computation}
\label{sec:qsmpc}
We present in this section an extension of the SDQC
Protocol~\ref{proto:dev_detect} from
Section~\ref{sec:dummyless_verification} based on the trappified
schemes in the $X-Y$ plane. We consider here that $n$
Clients want to perform a joint MBQC computation on private classical
inputs, receiving at the end either the same classical output or an
abort message. There are two steps in the SDQC protocol which must be
modified: the preparation of a state which is compatible with the SDQC
protocol and does not leak any information to coalitions of malicious
parties, and the classical interaction between with the server to
drive the computation and tests. If these components are available,
the composable security of the SDQC protocol ensures that the
multi-party version is also secure.

The second step is purely classical once the state and computation
have been fixed and we will use a Classical SMPC Resource to handle
it. This Resource will also sample the trappified canvas and embed the
Client's desired computation into it. Hence, no malicious coalition
will be able to learn where the tests are located among the blind
computations. The first step will make use of the Collaborative RSP
Protocol~\ref{proto:col_state_prep} from the previous section,
replacing the Orchestrator by calls to the Classical SMPC
Resource. The $n$ Clients will use it to prepare rotated $\ket{+}$
states on the Server such that the encryption angle $\theta$ is
unknown to any malicious coalition, which protects the blindness of
each computation.

Our resulting Secure Delegated Quantum Secure Multi-Party Computation
Protocol with Classical IO (Protocol~\ref{proto:mpqc}) is therefore an
information-theoretic upgrade of the Classical SMPC
functionality. This is the best one can hope for without an honest
majority since it is impossible in that case to construct an
information-theoretically secure Quantum SMPC protocol. Crucially, no
additional computational assumptions are used beyond what is required
to construct the Classical SMPC Resource. This modularity means that
we can instantiate our protocol using any post-quantum secure
assumption which is capable of constructing a Classical SMPC.

\paragraph{Quantum Secure Multi-Party Computation Resource.} 
Our protocol will construct the following Quantum Secure Multi-Party
Computation Resource~\ref{res:QSMPC}. It has $n + 1$ interfaces, one
for each Player and the last one for an Eavesdropper.  It allows $n$
Players to perform a collectively defined quantum computation $\cptp
C$ over their private classical inputs with the guarantee that their
computation is either executed properly, in which case Player $j$
receives the correct classical output, or it is aborted altogether.
It is allowed to leak a known value $l_{\rho}$ about the Players'
computation and input on the Eavesdropper's filtered interface.

\begin{resource}[ht]
\caption{Quantum Secure Multi-Party Computation with Classical IO}
\label{res:QSMPC}
\begin{algorithmic}[0]

\STATE \textbf{Inputs:}
\begin{itemize}
\item Player $j$ sends a classical bit-string $x_j$. It can also input two bits $f_j$ and $c_j$ as a filtered interface.
\item The $n$ Players send the classical description of a quantum polynomial-time computation $\cptp C$ with classical inputs and outputs.
\item The Eavesdropper can input two bits $e$ and $c$ as a filtered interface.
\end{itemize}
\STATE \textbf{Computation by the Resource:}
\begin{itemize}
\item If $e = 1$, the Resource sends the leakage $l_{\rho}$ to the Eavesdropper's interface.
\item If $c = 1$ or there exists $j$ such that $c_j = 1$, the Resource sends $\Abort$ to all Players $j$ such that $c_j = 0$.
\item It computes $O = \cptp C(x)$, where $x$ is the concatenation of strings $x_j$.
\item If there exists $j \in [n]$ such that $f_j = 1$, it sends $O$ to Player $j$.
\item If there has been no abort at this stage, it sends the outputs $O$ to all other Players $j$ in a similar fashion.
\end{itemize}

\end{algorithmic}
\end{resource}
In order to construct this resource, we will make use of its classical
equivalent. Our protocol will in the end be an information theoretical
upgrade of the following Resource.

\paragraph{Classical Secure Multi-Party Computation Resource.} 
Resource~\ref{res:SMPC} allows $n$ Players to provide their private
inputs and perform a collectively defined computation $C$ on them with
the guarantee that the computation is performed properly.  We assume
that it keeps an internal state between calls.

\begin{resource}[ht]
\caption{Classical Secure Multi-Party Computation}
\label{res:SMPC}
\begin{algorithmic}[0]
\STATE \textbf{Inputs:}
\begin{itemize}
\item Player $j$ sends a classical bit-string $x_j$. It can also input two bits $f_j$ and $c_j$ as a filtered interface.
\item The $n$ Players send the description of a classical polynomial-time computation $C$.
\end{itemize}

\STATE \textbf{Computation by the Resource:}
\begin{itemize}
\item If there exists $j$ such that $c_j = 1$, the Resource sends $\Abort$ to all Players $j$ such that $c_j = 0$.
\item It computes $O = C(x)$, where $x$ is the concatenation of strings $x_j$.
\item If there exists $j \in [n]$ such that $f_j = 1$, it sends $O$ to Player $j$.
\item If there has been no abort at this stage, it sends the outputs $O$ to all other Players $j$ in a similar fashion.
\end{itemize}
\end{algorithmic}
\end{resource}

\paragraph{Delegated QSMPC Protocol.}
Our final protocol will be built upon the two presented earlier. In an
execution the Trappified Delegated Blind Computation
Protocol~\ref{proto:dev_detect}, the Client can perform all of its
classical interactions with the Server via a Classical SMPC
Resource~\ref{res:SMPC} if it provides this resource with its input
and computation (angles and flow). This resource is then responsible
for sampling all the secret parameters -- angles, bits, order of test
and computation runs, which tests to perform -- and simply instructs
the Client to prepare specific states to send to the Server. Since
only rotated $\ket{+}$ states are required for this verification
protocol, this step can further be replaced by an instance of the
Remote State Preparation Resource~\ref{res:rsp} for states
$\{\ket{+_\theta}\}_{\theta \in \Theta}$, as sending a state from this
set is a perfect protocol constructing the RSP Resource. We can then
finally replace this resource by the Collaborative Remote State
Preparation Protocol~\ref{proto:col_state_prep}, in which the
Orchestrator is played by the Classical SMPC Resource.

In essence, the Classical SMPC together with the Collaborative RSP
emulate the behaviour of the honest Client in an execution of the
Trappified Delegated Blind Computation Protocol, whose tests --
described in Section~\ref{subsec:concrete-tests} -- needed to be
tailored specifically to require only the preparation of rotated
$\ket{+}$ states. The full description is given below in
Protocol~\ref{proto:mpqc}. We continue to refer to the Classical SMPC
Resource as the Orchestrator for simplicity, since in the Abstract
Cryptography framework there is no formal difference between an honest
party and an interactive Resource.

\begin{protocol}[ht]
\caption{Secure Delegated Quantum Secure Multi-Party Computation with Classical IO}
\label{proto:mpqc}
\begin{algorithmic} [0]

\STATE \textbf{Public Information:} 
\begin{itemize}
\item $G = (V, E, I, O)$, a graph with input and output vertices $I$ and $O$ respectively;
\item $\{I_j\}_{j \in [n]}$, a partition of the input vertices, with each $I_j$ being associated to Client $j$.
\item $\sch P$, a trappified scheme on graph $G$;
\item $\preceq_G$, a partial order on the set $V$ of vertices;
\item $N, d, w$, parameters representing the number of runs, the number of computation runs, and the number of tolerated failed tests.
\end{itemize}
    
\STATE \textbf{Clients' Inputs:}
\begin{itemize}
\item Each Client $j$ has as input a classical bit-string $x_j \in \bin^{\abs{I_j}}$.
\item The $n$ Clients collaboratively have as input a set of angles $\{\phi_i\}_{i \in V}$ and a flow $f$ which induces an ordering compatible with $\preceq_G$.
\end{itemize}

\STATE \textbf{Protocol:}
\begin{enumerate}
\item The Clients send their input $x_j$ to the Orchestrator, together with the computation angles $\{\phi_i\}_{i \in V}$ and flow $f$. Let $x$ be the concatenation of all $x_j$.
\item The Orchestrator and the Server perform an execution of the Trappified Delegated Blind Computation Protocol~\ref{proto:dev_detect}. Instead of having the Orchestrator send rotated states during the UBQC execution, they perform for each state an instance of the Collaborative State Preparation Protocol~\ref{proto:col_state_prep} together with the $n$ Clients.
	\begin{enumerate}
    \item The Orchestrator samples uniformly at random a subset $C \subset [N]$ of size $d$ representing the computation runs.
    \item For $k \in [N]$:
	\begin{enumerate}
		\item If $k \in C$, the Orchestrator sets the computation for the run to $(\{\phi_i\}_{i \in V}, f)$ with input $x$. Otherwise, the Orchestrator samples a test $(T, \sigma, \tau)$ from the trappified scheme $\sch P$.
		\item The Orchestrator and Server execute the chosen run with the UBQC Protocol~\ref{prot:UBQC}. For each qubit sent during the execution of the protocol, they instead execute the Collaborative RSP Protocol~\ref{proto:col_state_prep} together with the $n$ Clients.
		\item If the run is a test, the Orchestrator checks whether it passed.
	\end{enumerate}
	\item If the number of failed tests is greater than $w$, the Orchestrator sets the output to $(\bot,\rej)$.
	\item Otherwise, let $O$ be the majority vote on the output results of the computation runs. The Orchestrator sets the output to $(O, \acc)$.
	\end{enumerate}
\item The Orchestrator sends its set output to all Clients.
\end{enumerate}

\end{algorithmic}
\end{protocol}

\paragraph{Extending the Functionality.}
The presentation above restricts how the input and output are treated for simplicity's sake and any additional efficient classical pre- and post-processing steps can be performed by the Orchestrator with no impact on the security of the protocol.

\paragraph{Removing the Correction in the Collaborative RSP Protocol used with UBQC.}
The final step of the Collaborative RSP Protocol calls for the Orchestrator to instruct the Server to apply a correction $\X^b\Z((-1)^b\theta - \theta')$ to a state which in the honest case is equal to $\ket{+_{\theta'}}$, for a random value of $b \in_R \bin$ and the Orchestrator's desired angle $\theta$. This is required to make the protocol simulatable against a malicious coalition -- otherwise, the Simulator has no way of transmitting the correct state to the Server. However, in Protocol~\ref{proto:mpqc} these qubits are used in an execution of the UBQC Protocol, in which the Orchestrator requests that the Server measures the qubit in the basis $\{\ket{\pm_\delta}\}$ for $\delta = \phi' + \theta + r\pi$. Together, the unitary operations on this qubit in the honest case can be written as

\[
\Z(-\delta)\cptp E \X^b \Z((-1)^b\theta - \theta') \Z(\theta') \ket{+} \otimes \ket{\psi}
\]

for a state $\ket{\psi}$ representing the rest of the state and the graph entangling operation $\cptp E$. Then, this is equal to

\[
\Z(- \phi' - \theta - r\pi)\cptp E \Z(\theta - (-1)^b\theta') \Z((-1)^b\theta') \ket{+} \otimes \ket{\psi} = \Z(- \phi' - (-1)^b\theta' - r\pi)\cptp E \Z((-1)^b\theta') \ket{+} \otimes \ket{\psi}.
\]

By performing the change of variables $\hat{\theta} = (-1)^b\theta'$, which is drawn from the same distribution, we recover the state in the original UBQC Protocol, with no correction from the Orchestrator:

\[
\Z(- \phi' - \hat{\theta} - r\pi)\cptp E \Z(\hat{\theta}) \ket{+} \otimes \ket{\psi}.
\]

Therefore in the full protocol, requesting and applying the correction are unnecessary steps, either for correctness or security, since the states with or without these corrections are equal.

\paragraph{Security of QSMPC.}
We now prove the correctness and security of our QSMPC protocol using the composition of AC resources and protocols.

\begin{theorem}[Security of Delegated Quantum SMPC]
Suppose that the Trappified Delegated Blind Computation Protocol~\ref{proto:dev_detect} $\epsilon_V$-constructs the Secure Delegated Quantum Computation with Classical IO Resource~\ref{res:sdqc} for leak $l_\rho$. Then Protocol~\ref{proto:mpqc} $\epsilon_V$-constructs the Quantum Secure Multi-Party Computation with Classical IO Resource~\ref{res:QSMPC} from an interactive Classical Secure Multi-Party Computation Resource~\ref{res:SMPC} for the same leak $l_\rho$, against malicious coalitions that include at most the Server and $n-1$ Clients.
\end{theorem}

\begin{proof}
This proof is very simple and works by retracing in reverse order the high-level description of the protocol in the worst case with $n-1$ malicious Clients in collusion with a malicious Server

We first use the security of the Collaborative RSP Protocol as expressed in Theorem~\ref{thm:coll-sec} to replace each instance of this protocol with a call to the RSP Resource~\ref{res:rsp}, at no security cost. The Secure Classical Channel Resources from the Clients to the Orchestrator come for free since this party is now replaced by the Classical SMPC Resource in our protocol.

We can then replace these Resources with a direct quantum communication channel between the Orchestrator and the Server, since this protocol perfectly implements the RSP Resource. We obtain as a result exactly an execution of the UBQC Protocol~\ref{prot:UBQC} between the Orchestrator and the Server in step 2.b.ii of Protocol~\ref{proto:mpqc}. The whole step 2 of Protocol~\ref{proto:mpqc} is then exactly an execution of Protocol~\ref{proto:dev_detect} between the Orchestrator and the Server.

We then use the fact that this protocol $\epsilon_V$-constructs the Secure Delegated Quantum Computation with Classical IO Resource and replace it by a call to that resource with a cost of $\epsilon_V$.

In this final stage, the Clients send their desired computation and inputs to the Orchestrator, which only forwards the concatenated input to the SDQC Resource. This Resource leaks the value $l_\rho$ to the Server and returns the correct value to the Orchestrator if there has been no abort from the Server. The Orchestrator then sends back this output to the malicious Clients if they desire to receive it first. If there has been no abort at this stage, the Orchestrator finally transmits the output to the honest Clients as well. Therefore merging the Orchestrator -- a Classical SMPC Resource -- and the SDQC Resource yields exactly the behaviour of the desired QSMPC Resource between the $n$ Clients and the Server.
\end{proof}

\section{Discussion}
\label{sec:discussion}
 
\subsection{Comparison with Other QSMPC Protocols}
\label{sec:comparison}
Table \ref{tab:comp} below gives a comparison of our protocol with the peer-to-peer protocols of~\cite{DGJM+20:secure} and \cite{LRW20:secure}, and with the more recent semi-delegated protocol of~\cite{ACC20:multi}.
We note $n$ is the number of parties, $d$ the depth of the computation (MBQC for our paper, circuit for \cite{LRW20:secure} and $\{\T, \CNOT\}$-depth for \cite{DGJM+20:secure}), $t$ the number of $\T$ gates, $c$ the number of $\CNOT$ gates, $C_{\mathit{dist}}$ the code distance used in \cite{LRW20:secure} and $\eta$ a statistical security parameter. The values below correspond to the simple case where each player has a single qubit of input.

\begin{table}[htb!]
\centering
\begin{adjustbox}{max width=\textwidth, center}
\begin{tabular}{ccccc}
\toprule
 &
  \bfseries Dulek et al.~\cite{DGJM+20:secure} &
  \bfseries Lipinska et al.~\cite{LRW20:secure} &
  \bfseries Alon et al.~\cite{ACC20:multi} &
  \bfseries This work \\ \midrule
\textit{Security} &
  Stat. upgrade of CSMPC &
  Stat. &
  Comp. (FHE + CSMPC) &
  Stat. upgrade of CSMPC \\[8pt] \textit{Abort} &
  Unanimous &
  Unanimous &
  Identifiable &
  Unanimous \\[8pt] \textit{Composability} &
  Composable &
  Stand-Alone &
  Stand-Alone &
  Composable \\[8pt] \textit{Max adversaries} &
  $n - 1$ &
  $\floor{\frac{C_\mathit{dist} - 1}{2}}$ &
  $n - 1$ &
  $n - 1$ \\[8pt] \textit{Protocol nature} &
  Symmetric &
  Symmetric &
  Semi-Delegated &
  Delegated \\[8pt] \textit{Network topology} &
  Q and C: Complete &
  Q and C: Complete &
  Q and C: Complete &
  Q: Star / C: Complete \\[8pt] \textit{Q operations} &
  FTQC &
  FTQC &
  FTQC &
  \begin{tabular}[c]{@{}c@{}}Cl: Single Qubit\\ S: FTQC\end{tabular} \\[8pt] \textit{Classical SMPC} &
  \begin{tabular}[c]{@{}c@{}}Clifford Computation,\\ Operations in $\mathbb{Z}_2$, CT\end{tabular} &
  CT &
  \begin{tabular}[c]{@{}c@{}}Clifford Computation,\\ FHE verification\end{tabular} &
  Operations in $\mathbb{Z}_8$, $\mathbb{Z}_2$, CT \\[12pt] \textit{Rounds (C)} &
  $\order{d+\eta(N+t)}$ &
  $d + 2$ &
  $\order{1}$ &
  $d + 3$ \\[8pt] \textit{Rounds (Q)} &
  \begin{tabular}[c]{@{}c@{}}Par: $\order{nd}$\\ Seq: $\order{n(n+t+c)}$\end{tabular} &
  \begin{tabular}[c]{@{}c@{}}Par: $3$ ($2$ if C output)\\ Seq: $\order{\eta^2 (n+t)}$\end{tabular} &
  Par: $\order{n^4}$ &
  \begin{tabular}[c]{@{}c@{}}Par: $1$\\ Seq: $\order{\eta nd}$\end{tabular} \\[12pt] \textit{Size of Q memory} &
  \begin{tabular}[c]{@{}c@{}}Par: $\order{\eta^2(n+t))}$\\ Seq: $\order{\eta^2 n}$\end{tabular} &
  \begin{tabular}[c]{@{}c@{}}Par: $\order{\eta^2 n(n+t)}$\\ Seq: $\order{n^2}$\end{tabular} &
  Par: $\order{tn^9\eta^2}$ &
  \begin{tabular}[c]{@{}c@{}}Cl: $0$\\ S (par): $\order{\eta n^2d}$\\ S (seq): $\order{nd}$\end{tabular} \\ \bottomrule
\end{tabular}
\end{adjustbox}

\vspace*{0.2cm}
\caption{Comparison with \cite{DGJM+20:secure,LRW20:secure,ACC20:multi}. Q stands for quantum and C for classical. The abbreviations Cl and S stand for Client and Server respectively. Stat. means statistical, FTQC stands for Fault-Tolerant Quantum Computer and CT for Coin-Toss.}
\label{tab:comp}
\end{table}

\paragraph{Security guarantees.}
Reference \cite{DGJM+20:secure} achieve an information-theoretic upgrade of a Classical SMPC to the quantum domain, secure against an arbitrary number of corrupted parties. On the other hand, the protocol from \cite{ACC20:multi} is only computationally-secure since it relies on a Fully-Homomorphic Encryption Scheme on top of the Classical SMPC, but it is also secure against arbitrary corruptions. The protocol of \cite{LRW20:secure} constructs an information-theoretically secure Quantum SMPC but suffers from an artificial blow-up in the number of participants and exchanged qubits.\footnote{It is based on error-correcting codes and the size of the code must correspond to the number of players $n$. The maximum number of cheaters tolerated by the protocol is the number of correctable errors $\floor{\frac{C_{\mathit{dist}} - 1}{2}}$, which by the quantum Singleton bound \cite{R99:codes} is at most $\floor{\frac{n-1}{4}}$. In their example, $7$ players are required for implementing a two-party computation since the code that is used is of size $7$ and corrects $1$ error. This leads to a situation where $5$ participants that don't have inputs nor outputs must still exchange messages and none can be malicious if one of the players with inputs is.} The protocols of \cite{LRW20:secure,ACC20:multi} are proven secure in the Stand-Alone Model, whereas ours and that of \cite{DGJM+20:secure} are fully composable. On top of blindness, all protocols provide verifiability with unanimous abort apart from that of \cite{ACC20:multi} which achieves the stronger notion of identifiable abort.\footnote{A protocol satisfies the unanimous abort property if all honest players abort at the same time, as compared with selective aborts where the Adversary can choose which players will abort separately. On top of that, identifiable abort means that all honest players agree on the malicious party responsible for the failure of the protocol.}

\paragraph{Communication requirements.}
One key advantage of our protocol over the others lies in its delegated nature, where only one participant needs a full fault-tolerant quantum computer while the rest only perform very limited quantum operations, compared with the symmetric setup in \cite{DGJM+20:secure,LRW20:secure} where all participant has requires fault-tolerance. The protocol of \cite{ACC20:multi} can be considered semi-delegated in the sense that the brunt of the quantum computation is performed by a single player. However, all players must have the ability to perform arbitrary Cliffords on large states and cannot do so without having at their disposal a full fault-tolerant quantum computer. This is also reflected in the network topology: whereas the best performance in \cite{DGJM+20:secure,LRW20:secure,ACC20:multi} can only be reached by using a complete quantum and classical communication graph, we only need a star graph for quantum communications. While the network topology of \cite{DGJM+20:secure} and \cite{LRW20:secure} can also be star-shaped -- with one player acting as a router -- this would degrade their performance in terms of quantum communication rounds.

\paragraph{Usage of Classical Primitives.}
Regarding classical primitives, \cite{LRW20:secure} only requires secure coin-tossing and authenticated broadcast channels (information-theoretically secure since they can rely on an honest majority). We only use our Classical SMPC to perform coin-tossing, basic string operations (array lookup) and computations in $\mathbb{Z}_8$ and $\mathbb{Z}_2$. 
The Classical SMPC is more complex in \cite{DGJM+20:secure,ACC20:multi} since it must be able to sample uniformly at random and perform computations on the classical descriptions of arbitrary Cliffords.

\paragraph{Rounds of communication.}
We can now quantify more precisely the number of classical rounds of communication or calls to the Classical SMPC resource, quantum rounds of communication, and size of quantum memory required by each participant in the protocol. \cite{DGJM+20:secure} calls the Classical SMPC very often: a constant number of times for each input qubit and gate in the circuit. But the most costly part is the generation of ancillary magic states (for implementing $\T$ gates via gate-teleportation), which requires $\order{\eta(n+t)}$ invocations of the Classical SMPC. Our protocol simply uses $d + 3$ calls to this Resource, $2$ for setting up the state and $2$ for the key-release step. This is equivalent to the classical communication requirements of \cite{LRW20:secure}, where they only need $d + 2$ classical broadcasts per participant (with one for setting up the shared randomness and another for the state preparation). If all quantum communications are done in parallel in \cite{LRW20:secure}, it can be further parallelised to only require a constant number of classical broadcast rounds. The protocol of \cite{ACC20:multi} uses FHE (classical and quantum) to perform the computation and consequently the number of calls to the Classicl SMPC is only constant. We note that using a classical primitive called functional encryption, where a party in possession of an evaluation key can recover the clear-text of a function of the encrypted values (and only that), would allow to attain the same result for our construction by allowing the Server to compute the next measurement angle as a function of the encrypted secrets and previous measurement results.

The protocol of \cite{DGJM+20:secure} requires numerous rounds of quantum communication as they need to send encoded states around for the verification of inputs and $\T$ and $\CNOT$ gates. After parallelisation the total cost is $\order{nd}$ quantum rounds. 
\cite{ACC20:multi} aims to remove the circuit dependency in the number of rounds, obtaining $\order{n^4}$ quantum rounds in the worst case in the case where the protocol is parallelised.\footnote{They send states along a path of size $n^2$ in the communication graph of the parties, and remove a party if it doesn't deliver a packet before resending the states along a different path of the same size. In the worst case where there are $n-1$ malicious players which do not want to get caught cheating, they can drop $(n-1)(n-2)/2$ packets without being disconnected from the communication graph.}

\paragraph{Qubit Count and Memory Requirements.}
\cite{LRW20:secure} seeks to optimise the quantum memory requirement of players and therefore their communication is done sequentially, yielding $\order{\eta^2(n + t)}$ quantum rounds. Parallelisation lowers it to $3$ (or $2$ for classical outputs), at a higher quantum memory cost for all parties. 
Our protocol is optimal as there is only a single quantum round (in the parallel case): sending to the Server all states required for the collaborative state preparation phase.

Finally, the number of qubits required by \cite{DGJM+20:secure} during the computation phase is $\order{\eta(n+t)}$ for each participant (they encode each of their input qubits, ancillae and magic states using $\order{\eta}$ qubits). However they use $\order{\eta^2(n+t)}$ additional qubits in the offline phase to prepare the ancillary qubits (if the quantum communications are performed in parallel). On the other hand, \cite{LRW20:secure} reduces the number of qubits for each participant to $\order{n^2}$ for sequential quantum communication, but this blows up to $\order{\eta^2 n(n+t)}$ if parallelised. The construction from \cite{ACC20:multi} uses a compiler that adds automatically a cost of $\order{n^2}$ for each base qubit. The costly double encryptions and multiple layers of traps, in particular for the magic state distillation procedure, yields a total quantum memory cost per participant of at least $\order{tn^9\eta^2}$ (this is a weak lower bound). In our paper the Server needs $\order{nd}$ qubits to perform each blind computation or test. Each qubit in these graphs is generated using $n$ qubits via the Collaborative RSP Protocol and the computations and tests are repeated $\order{\eta}$ times in total, resulting in a total qubit cost of $\order{\eta n^2d}$ for parallelised quantum communication but only $\order{nd}$ if the rounds are performed sequentially. However, the Clients can prepare these states on the fly and the Clients do not need quantum memory.

\subsection{Impossibility of Single-Qubit Privacy Amplification on the Whole Bloch Sphere}\label{sec:impossibility}
The construction and security of Protocol~\ref{proto:mpqc} relies on the
composition of a collaborative encryption gadget with the regular
robust VBQC protocol driven by the Orchestrator.

The crucial features of the collaborative encryption are that (i) a
single honest Client providing a random state from the allowed input
set is enough to randomize the output of the gadget, and (ii) no
information about the state provided by an honest Client leaks to the
Server. These are the two properties that were shown to hold
in \S~\ref{subsec:coll-st-prep}.

It can also be seen that those do not hold whenever the set of input
states for the clients not only comprise the 8 $\ket{+_\theta}$ states
in the $X-Y$ plane but also the computational basis states $\ket
0, \ket 1$. The reason is that in such case, if the central qubit is
set to a computational basis state, it cannot be randomized by the
states provided by other clients.

While this specific failure is contingent to the chosen transformation
implemented by our gadget, we will show here that it is indeed a more
generic problem that gadgets fulfilling (i-ii) have in common, thereby
restricting these ``gadget-assisted'' approaches to verification of
classical input classical output computations.

First, we give a mathematical definition of (i):
\begin{definition}[Randomizing gadget]
  Let $P$ be a protocol with two Clients and one Server such that it
  takes a quantum state at each Client's input interface, and produces a
  quantum state at the Server's output interface together with a common
  classical bit string at each of the parties output interface.

  We say that this gadget is \emph{randomizing} whenever conditionned
  on the value of the common bit string, the linear maps implemented
  by the protocol when one of the two input states is fixed is
  invertible for pure input states.
\end{definition}

The motivation for this definition is simple: whenever one of the
input is fixed, then the other one is enough to randomize the output
at the Server's side. The role of the common bit string shared by all
the parties at the end of the protocol is to allow the possibility of
having a linear map that depends on this bit string as it is the case
in our construction. As a consequence, the output state at the
Server's interface might not be normalized in order to encapsulate the
probability of a specific common bit string to be produced by the
protocol.

The following Lemma~\ref{lemma:no_linear_map} shows that for a fixed
common string, there will always exist a specific state for one of the
two inputs such that the map will not be invertible. This can result
in one of the two following cases. Either the output is a fixed
non-zero quantum state or it produces the null vector. In the former,
this implies that the gadget is not able to correctly produce random
states required by the VBQC protocol to be secure. In the latter,
observing a specific common bit string excludes some input state for
the honest Client, thereby also violating the assumptions required to
obtain the security of the whole protocol.

\begin{lemma}\label{lemma:singular_matrix}
  Every two-dimensional linear subspace $V \subset \mathbb{C}^{2\times
  2}$ contains at least one nonzero, singular matrix.
\end{lemma}

\begin{proof}
  Let $A, B \in V$ form a basis of $V$. If $A$ or $B$ is singular, the
  claim is trivial. Assume henceforth that both $A$ and $B$ are
  invertible.
	
  For $\alpha \in \mathbb{C}$, let $C_\alpha = A + \alpha B$. Clearly,
  $C_\alpha \in \operatorname{span}(\{A,B\})$. Since $A$ and $B$ are
  linearly independent, $C_\alpha \neq 0$. It further holds that

  \begin{align*}
    \det (C_\alpha)
    & = \det
    \begin{bmatrix}
      a_{11} + \alpha b_{11} & a_{12} + \alpha b_{12} \\
      a_{21} + \alpha b_{21} & a_{22} + \alpha b_{22}
    \end{bmatrix}\\
    & = ( a_{11} + \alpha b_{11} ) ( a_{22} + \alpha b_{22} ) - ( a_{12} + \alpha b_{12} ) ( a_{21} + \alpha b_{21} ) \\
    & = \alpha^2 ( b_{11} b_{22} - b_{12} b_{21} ) + \alpha ( a_{11} b_{22} + b_{11} a_{22} - a_{12} b_{21} - b_{12} a_{21} ) + a_{11} a_{22} - a_{12} a_{21} \\
    & = \alpha^2 \det(B) + \alpha ( a_{11} b_{22} + b_{11} a_{22} - a_{12} b_{21} - b_{12} a_{21} ) + \det(A).
\end{align*}
  As $\det(B) \neq 0$, this is a polynomial of degree 2 in the
  variable $\alpha$. By the fundamental theorem of algebra, this
  polynomial admits at least one complex root.
\end{proof}

\begin{lemma}\label{lemma:no_linear_map}
  There exists no linear map $\Xi : \mathbb{C}^{2 \times
  2} \to \mathbb{C}^2$ such that for all nonzero $v \in \mathbb{C}^2$
  both $\Xi(\cdot \otimes v)$ and $\Xi(v \otimes \cdot)$ are
  invertible.
\end{lemma}
\begin{proof}
  Assume the existence of such a map $\Xi$. By the rank-nullity theorem,
  it holds then that
  \begin{align*}
    \dim(\operatorname{Ker}(\Xi)) = \dim(\mathbb{C}^{2 \times 2}) - \dim(\operatorname{Im}(\Xi)) \geq 2.
  \end{align*}
  By Lemma~\ref{lemma:singular_matrix}, there exists a rank-one matrix
  $C\in\operatorname{Ker}(\Xi)$. We can rewrite $C = vw^T = v \otimes
  w$ with nonzero vectors $v,w \in \mathbb{C}^2$. It follows that $\Xi
  (v \otimes w) = 0$ which contradicts the invertibility of
  $\Xi(\cdot \otimes w)$ and $\Xi(v \otimes \cdot)$.
\end{proof}

This leads us to conclude that such gadget assisted approaches will
inherently be limited to classical I/O computations.

\subsection{Open Questions}
This work closes a gap between the circuit and MBQC models regarding
secure multi-party computations. It shows that both are able to
perform the required lift from classical to quantum in a statistically
secure way, in spite of the more stringent requirements the delegation
imposes on what clients can do. Yet, this is only partially
satisfactory as we do not consider the quantum input/output case. This
specific question was considered by some of the authors. This lead to
designing a protocol that was similar in spirit to the one presented
here, but where the Collaborative Remote State Preparation would not
only be able to prepare states in the $X-Y$ plane, but also dummy
qubits. An attack on this protocol is analysed
in~\S~\ref{sec:post-mortem}. It's discovery initiated the current work
using dummyless verification as a way to avoid it. Yet, we also show in
\S~\ref{sec:impossibility} that such approach based on Collaborative
Remote State Preparation outside a single plane is not likely to
succeed, thereby leaving open the question of how to perform Delegated QSMPC with quantum I/O.

Other open questions regarding SMPC in the MBQC model include the
verification of sampling with possibly better than polynomial security
bounds. The question of the delegation of fault-tolerant computation
in the MBQC model is also a long standing open question that we
believe can benefit from the theoretical tools developed in
~\cite{KKLM+22:framework} and from an approach similar to the one
exemplified in this work.

Finally, \cite{MKAC+22:qenclave} showed how to blindly delegate quantum computations with trusted rotations, even if both state preparations and measurements are untrusted, but left open the question whether verification is possible in this setting. The difficulty of verification seems to stem from the fact that (i) their analysis concerns states in the $X-Y$ plane, but not dummy states, and (ii) the remotely prepared states are blind, but not necessarily verifiable. While this work does not overcome the second obstacle, it shows that verification is indeed possible without the remote preparation of dummy states, and therefore constitutes a step towards the solution of this open problem.

\paragraph*{Acknowledgments.}
The authors would like to thank Michael Oliveira for discussions about
the results from Section~\ref{sec:impossibility}. TK was supported by a Leverhulme Early Career Fellowship.
HO was partially funded by the Hybdrid HPC Quantum Initiative.
EK, DL, and HO acknowledge the support by ANR research grant ANR-21-CE47-0014 (SecNISQ).
LM is grateful for support from the grant BPI France Concours Innovation PIA3 projects DOS0148634/00 and DOS0148633/00 – Reconfigurable Optical Quantum Computing.

\bibliographystyle{splncs04}
\bibliography{../qubib/qubib.bib}

\appendix
\part*{Auxiliary Supporting Material}

\section{The Abstract Cryptography Framework}
\label{subsec:ac}
Abstract Cryptography is a framework for defining and proving the security of cryptographic protocols, first introduced in \cite{MR11:abstract-cryptography,M12:constructive-cryptography}.  
Its main advantage compared to so-called Stand-Alone Models such as~\cite{HSS15:classical} is that any system that follows the structure defined by the framework is inherently composable, in the sense that if two protocols are secure separately, the framework guarantees at an abstract level that their sequential or parallel composition is also secure. It is equivalent to the Quantum Universal Composability (Q-UC) Model of~\cite{U10:universally} if a single Adversary controls all corrupted parties -- which is the case in this work. Therefore any protocol which is secure in the Q-UC model is also secure in the AC model considered here. 
We refer the reader to \cite{DFPR14:composable} for a more in-depth presentation.

In this framework, the purpose of a secure protocol $\pi$ is, given a number of available resources $\mathcal{R}$, to construct a new resource -- written as $\pi \mathcal{R}$.  
This new resource can be itself reused in a future protocol.

The actions of all honest players in a given protocol are represented as a sequence of efficient CPTP maps acting on their internal quantum registers -- which may contain communication registers, both classical and quantum.  
An $n$-party quantum protocol is therefore described by $\pi = (\pi_1, \ldots, \pi_n)$ where $\pi_j$ is the aforementioned sequence of efficient CPTP maps executed by party $j$, called the \emph{converter} of party $j$.

A \emph{resource} $\mathcal{R}$ is described as a sequence of CPTP maps with an internal state.  
It has input and output interfaces describing which party may exchange states with it.  
Some interfaces may be filtered, meaning that they are only accessible to a corrupted party.\footnote{In this paper filtered input interfaces consist of single bits, set to $0$ in the default, honest case.}
It works by having the party sending it a given state at one of its input interfaces, applying the specified CPTP map after all input interfaces have been initialised and then outputting the resulting state at its output interfaces in a specified order.  
Classical resources are modelled by considering that the input state is measured in the computational basis upon reception and the output is a measurement result on its internal state.

In order to define the security of a protocol, we need to give a pseudo-metric on the space of resources.  
The security analysis then consists of considering a special type of converters called \emph{distinguishers}.
The distinguisher's aim is to discriminate between resources $\mathcal{R}_0$ and $\mathcal{R}_1$ which have the same input and output interfaces.
It attaches to the inputs and outputs of one of the resources, interacting with it according to its own -- possibly adaptive -- strategy, and outputs a single bit indicating its guess as to which resource it had access to.  
Two resources are said to be indistinguishable if no distinguisher can make this guess with good probability. 

\begin{definition}[Indistinguishability of Resources]
\label{def:indist-res}
Let $\epsilon(\eta)$ be a function of security parameter $\eta$ and $\mathcal{R}_0$ and $\mathcal{R}_1$ be two resources with same input and output interfaces.  
The resources are $\epsilon$-statistically-indistinguish\-able if, for all distinguishers $\mathcal{D}$, we have:
\begin{equation} \label{eq:dist}
\Bigl\lvert\Pr[b = 1 \mid b \leftarrow \mathcal{D}\mathcal{R}_0] - \Pr[b = 1 \mid b \leftarrow \mathcal{D}\mathcal{R}_1]\Bigr\rvert \leq \epsilon
\end{equation}
We then write $\mathcal{R}_0 \!\!\!\underset{\mathit{stat}, \epsilon}{\approx}\!\!\! \mathcal{R}_1$.
\end{definition}

The correctness of a protocol protocol $\pi$ applied to resource $\mathcal{R}$ can be expressed as the indistinguishability between the resource $\pi\mathcal{R}$ and a desired target resource $\mathcal{S}$.

The security of the protocol is captured by the fact that the resources remain indistinguishable if we allow some parties to deviate in the sense that they are no longer forced to use the converters defined in the protocol but can use any other CPTP maps instead.
This is done by removing the converters for those parties in Equation~\ref{eq:dist}, keeping only $\pi_{M^c} = \{\pi_j\}_{j \notin M}$ where $M$ is the set of corrupted parties. On the other side, there must exist a converter called a \emph{simulator} which attaches to the interfaces of $\mathcal{S}$ for corrupted parties $j \in M$ and aims to reproduce the transcript of honest players interacting with the corrupted ones.
The security is formalised as follows in Definition \ref{def:const-sec}.

\begin{definition}[Construction of Resources]
\label{def:const-sec}
Let $\epsilon(\eta)$ be a function of security parameter $\eta$.  
We say that an $n$-party protocol $\pi$ $\epsilon$-statistically-constructs resource $\mathcal{S}$ from resource $\mathcal{R}$ against adversarial patterns $\mathsf{P} \subseteq \wp([N])$ if:
\begin{enumerate}
\item It is correct: $\pi \mathcal{R} \!\!\!\underset{\mathit{stat}, \epsilon}{\approx}\!\!\! \mathcal{S}\bot$, where $\bot$ filters the malicious interfaces;
\item It is secure for all subsets of corrupted parties in the pattern $M \in \mathsf{P}$: there exists a converter called simulator $\sigma_M$ such that $\pi_{M^c}\mathcal{R} \!\!\!\underset{\mathit{stat}, \epsilon}{\approx}\!\!\! \mathcal{S} \sigma_M$.
\end{enumerate}
\end{definition}

We can now present the General Composition Theorem (Theorem 1 from~\cite{MR11:abstract-cryptography}).

\begin{theorem}[General Composition of Resources]
Let $\mathcal{R}$, $\mathcal{S}$ and $\mathcal{T}$ be resources, $\alpha$, $\beta$ and $\mathsf{id}$ be protocols (where protocol $\mathsf{id}$ does not modify the resource it is applied to). Let $\circ$ and $\mid$ denote respectively the sequential and parallel composition of protocols and resources. Then the following implications hold:
\begin{itemize}
\item The protocols are \emph{sequentially composable}: if $\alpha \mathcal{R} \!\!\!\underset{\mathit{stat}, \epsilon_{\alpha}}{\approx}\!\!\! \mathcal{S}$ and $\beta \mathcal{S} \!\!\!\underset{\mathit{stat}, \epsilon_{\beta}}{\approx}\!\!\! \mathcal{T}$ then $(\beta \circ \alpha) \mathcal{R} \!\!\!\underset{\mathit{stat}, \epsilon_{\alpha} + \epsilon_{\beta}}{\approx}\!\!\! \mathcal{T}$
\item The protocols are \emph{context-insensitive}: if $\alpha \mathcal{R} \!\!\!\underset{\mathit{stat}, \epsilon_{\alpha}}{\approx}\!\!\! \mathcal{S}$ then $(\alpha \mid \mathsf{id}) (\mathcal{R} \mid \mathcal{T}) \!\!\!\underset{\mathit{stat}, \epsilon_{\alpha}}{\approx}\!\!\! (\mathcal{S} \mid \mathcal{T})$
\end{itemize}
\end{theorem}

Combining the two properties presented above yields concurrent composability (the distinguishing advantage cumulates additively as well). 

The computational versions of these definitions are obtained by quantifying over quantum polynomial time parties.  
Composing a statistically-secure protocol with a computationally-secure protocol is possible provided that the simulator for the statistically-secure one runs in expected polynomial time.  
The resulting protocol is of course only computationally-secure.

\paragraph{Comments on the Security Framework.}
First, we always consider in this work a single Adversary controlling all the corrupted parties. As explained above, it is therefore possible to instantiate all purely classical Resources using any classical protocol which is secure in the Q-UC framework of~\cite{U10:universally} with the same security guarantees. It is also possible to instantiate them with any classical UC-secure protocol whose security relies on a quantum-hard problem thanks to Theorem $18$ (Quantum Lifting Theorem – Computational) from~\cite{U10:universally}.

Also, it is impossible to have fairness of output distribution in the case of a dishonest majority, the malicious parties can always choose to receive their output before the honest players. This is modelled in the resources by a filtered bit $f_j$ at each player's interface, indicating that it receives the output before others. The corrupted players can then decide to make the honest players abort before receiving their output. 
\section{Measurement-Based Quantum Computing}
\label{subsec:mbqc}
The protocols in the present paper relies on Measurement-Based Quantum Computing (or MBQC).
The MBQC model of computation emerged from the gate teleportation principle.
It was shown in \cite{RB01:one} that any quantum computation can be implemented by performing single-qubit measurements on a type of entangled states called graph states.

Given a graph $G = (V, E)$, and input and output vertices $I, O \subseteq V$, the corresponding graph state is generated by initialising a qubit in state $\ket{+}$ for each vertex in $V$ and performing entangling operator $\CZ$ between qubits whose vertices are linked by an edge in $E$. The qubits are measured according to an order given by a function $f: O^c \rightarrow I^c$ called the flow of the computation.

We define the rotation operator around the $Z$ axis of the Bloch sphere by an angle $\theta$ as $\Z({\theta}) = \begin{pmatrix} 1 & 0 \\ 0 & e^{i\theta} \end{pmatrix}$ and $\ket{+_{\theta}} = \Z(\theta)\ket{+} = \frac{1}{\sqrt{2}}(\ket{0} + e^{i\theta}\ket{1})$. 
For approximate universality, we can restrict the set of angles to $\Theta = \qty{\frac{k\pi}{4}}_{k \in \qty{0, \ldots, 7}}$ \cite{BFK09:universal}. The measurement associated to an angle $\phi \in \Theta$ is given by the basis $\ket{\pm_\phi}$. We consider in this paper that this measurement is performed by rotating the state to be measured using the operation $\Z(-\phi)$ and then measuring in the $X$-basis.

Later measurement may depend on the outcomes of previous measurements. Let $\{\phi(v)\}_{v \in O^c}$ be a set of default measurement angles for non-output qubits.
Let $S_X(v)$ and $S_Z(v)$ be respectively the $\X$ and $\Z$ dependency-sets for qubit $v$.\footnote{These sets are also given by the flow, see~\cite{HEB03:multi,DK06:determinism} for details.}
The measurement result $s(w)$ for qubit $w \in S_X(v) \cup S_Z(v)$ induces Pauli corrections on qubit $v$ which are equivalent to measuring qubit $v$ with corrected angle $\phi'(v) = (-1)^{s_X(v)} \phi(v) + \pi s_Z(v)$, where $s_X(v) = \bigoplus_{w \in S_X(v)} s(w)$ and $s_Z(v) = \bigoplus_{w \in S_Z(v)} s(w)$.

The special case of classical inputs is handled by adding an angle $x(v)\pi$ to the measurement angle $\phi(v)$ of input qubit $v \in I$. Classical outputs correspond to the case where all qubits are measured.

The classical input-output computation is defined by a graph $G = (V, E)$, input and output vertices $I, O \subseteq V$, a set of default measurement angles $\{\phi(v)\}_{v \in V}$ and a flow function $f: O^c \rightarrow I^c$.  
To perform the computation, one generates the graph state associated to $G$, performs the measurements with angles $\phi'(v)$ using the default angles and the flow. The outcome is defined by bit-string $\{s(v)\}_{v \in O}$. 
\subsection{Universal Blind Quantum Computing}
\label{subsec:ubqc}
An MBQC computation can easily be delegated by a Client to a quantum Server by having the Server create the graph state and perform measurements instructed by the Client. Universal Blind Quantum Computation (or UBQC)~\cite{BFK09:universal} is an upgrade of delegated MBQC which guarantees that the Server does not learn anything about the computation besides the computation graph, order of measurements and position of output qubits. 

This is achieved as follows. The Client hides the computation by sending rotated states $\ket{+_{\theta(v)}}$ instead of $\ket{+}$ for each qubit $v$.  
The effect of this rotation is cancelled by a corresponding rotation of the measurement angle.  
An additional parameter $r(v)$ adds an extra $r(v)\pi$ rotation to the measurement angle in order to hide the measurement outcome.  
This can be classically accounted for by the Client by setting $s(v) = r(v) \oplus b(v)$, where $s(v)$ is used as above in MBQC for corrections while $b(v)$ is the outcome returned by the Server.  
The measurement angle sent by the Client is then $\delta(v) = \phi'(v) + \theta(v) + r(v) \pi + x(v)\pi$, where $x$ is its input bit-string, extended by $0$ on non-input vertices.

We give the full UBQC protocol for classical inputs and outputs in Protocol \ref{prot:UBQC}. 

\begin{protocol}[ht]
\caption{Classical Input-Output Universal Blind Quantum Computation}
\label{prot:UBQC}
\begin{algorithmic} [0]

\STATE \textbf{Public Information:} 
\begin{itemize}
\item $G = (V, E, I, O)$, a graph with input and output vertices $I$ and $O$ respectively;
\item $\preceq_G$, a partial order on the set $V$ of vertices;
\end{itemize}

\STATE \textbf{Client's Inputs:} A bit-string $x \in \bin^{|I|}$ and the classical description of a unitary $U$ as default measurement angles $\{\phi(v)\}_{v \in V}$ and a flow $f$ which induces an ordering compatible with $\preceq_G$.

\STATE \textbf{The Protocol:}

\begin{enumerate}
\item The Client, for each vertex $v \in V$, samples at random $\theta(v) \in_R \Theta$ and sends $\ket{+_{\theta(v)}}$ to the Server.
\item The Server receives the qubits one-by-one and applies the entangling operations $\CZ$ that correspond to the edges of the graph $G$.
\item For each qubit $v \in V$, following the partial order of the flow:
\begin{enumerate}
\item The Client samples at random $r(v) \in_R \bin$, calculates and sends an angle to the Server:
\begin{equation}
\delta(v) = \phi'(v) + \theta(v) + r(v) \pi + x(v)\pi.\label{eq:ubqc_update}
\end{equation} 
\item The Server measures in the $\{\ket{+_{\delta(v)}},\ket{-_{\delta(v)}}\}$ basis and returns to the Client outcome $b(v)$. The Client sets $s(v) = b(v) \oplus r(v)$.
\end{enumerate}
\item The Client sets the bit-string $\{s(v)\}_{v \in O}$ as its output.
\end{enumerate}

\end{algorithmic}
\end{protocol} 
\section{Post-Mortem of Previous Protocol}
\label{sec:post-mortem}
A subset of the authors of the current paper proposed an earlier protocol for QSMPC \cite{KKMO21:delegating}. We show here the limits of the design and discuss possible paths towards fixing it.

\paragraph{State-Selective Flipping Attack.}

The principle of the previous protocol was to separate the computation
in two parts. The first section, which is blind only and not
verifiable, is responsible for preparing the verifiable graph state
from~\cite{KW17:optimised}, i.e.~a single graph state which includes
traps. This requires to prepare both rotated qubits and dummies. The
Collaborative RSP prepares only rotated states which must then either
be transformed into dummies $\{\ket{0}, \ket{1}\}$ or left undisturbed
(for computation and trap qubits). This is done with a blind
computation on all these qubits, the additional qubits required for
this computation being also generated using the Collaborative
RSP. This is essentially a way to extend the Collaborative RSP to a
bigger set of states.

The blindness of this gadget is proven in the Abstract Cryptography framework, so it would seem that it can be composed with the single-Client SDQC protocol to yield QSMPC in the same way as in the current work. However, this is not the case since we do not verify that the Server acts honestly so the final state is correct up to a global deviation. In general this deviation depends on the state that is being prepared, in particular the deviation can depend on whether the final state is a dummy or computation/trap qubit. Conscious of this, \cite{KKMO21:delegating} exhibit a number of sufficient conditions on the computation so that this global deviation is independent of the secret state the the Clients prepare collaboratively.

These conditions are as follows:
\begin{enumerate}
\item The inputs in the graph of the Clients' desired computation have degree $1$.
\item All measurement angles are from the set $\{k\pi/2\}_{0 \leq k \leq 3}$, i.e.~the computation is Clifford.
\item The graph, flow of this computation and the angles of all vertices beyond the first layer of the gadget are independent of the final desired state of the qubits.
\end{enumerate}

This final condition restrains a lot the possible types of computations that can be performed in this step but \cite{KKMO21:delegating} proposes a scheme which seems to satisfy them. We recall it here for completeness. The MBQC pattern is given by the following graph and angles.

\begin{center}
  \begin{picture}(0,0)\includegraphics{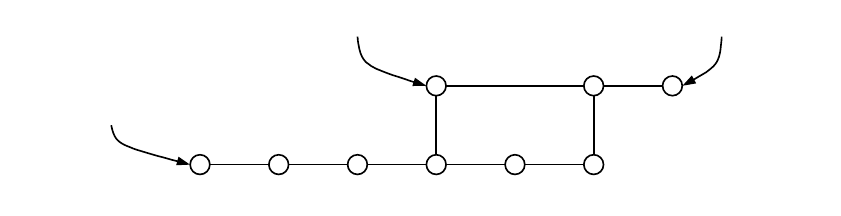}\end{picture}\setlength{\unitlength}{4144sp}\begingroup\makeatletter\ifx\SetFigFont\undefined \gdef\SetFigFont#1#2#3#4#5{\reset@font\fontsize{#1}{#2pt}\fontfamily{#3}\fontseries{#4}\fontshape{#5}\selectfont}\fi\endgroup \begin{picture}(3841,972)(-914,608)
\put(271,659){\makebox(0,0)[lb]{\smash{{\SetFigFont{8}{9.6}{\rmdefault}{\mddefault}{\updefault}{\color[rgb]{0,0,0}$\pi/2$}}}}}
\put(631,659){\makebox(0,0)[lb]{\smash{{\SetFigFont{8}{9.6}{\rmdefault}{\mddefault}{\updefault}{\color[rgb]{0,0,0}$\pi/2$}}}}}
\put(541,1469){\makebox(0,0)[lb]{\smash{{\SetFigFont{8}{9.6}{\rmdefault}{\mddefault}{\updefault}{\color[rgb]{0,0,0}Arbitrary input $\rho$}}}}}
\put(2251,1469){\makebox(0,0)[lb]{\smash{{\SetFigFont{8}{9.6}{\rmdefault}{\mddefault}{\updefault}{\color[rgb]{0,0,0}Output qubit}}}}}
\put(-44,659){\makebox(0,0)[lb]{\smash{{\SetFigFont{8}{9.6}{\rmdefault}{\mddefault}{\updefault}{\color[rgb]{0,0,0}$0$}}}}}
\put(1036,659){\makebox(0,0)[lb]{\smash{{\SetFigFont{8}{9.6}{\rmdefault}{\mddefault}{\updefault}{\color[rgb]{0,0,0}$0$}}}}}
\put(1396,659){\makebox(0,0)[lb]{\smash{{\SetFigFont{8}{9.6}{\rmdefault}{\mddefault}{\updefault}{\color[rgb]{0,0,0}$0$}}}}}
\put(1756,659){\makebox(0,0)[lb]{\smash{{\SetFigFont{8}{9.6}{\rmdefault}{\mddefault}{\updefault}{\color[rgb]{0,0,0}$0$}}}}}
\put(1756,1289){\makebox(0,0)[lb]{\smash{{\SetFigFont{8}{9.6}{\rmdefault}{\mddefault}{\updefault}{\color[rgb]{0,0,0}$0$}}}}}
\put(1036,1289){\makebox(0,0)[lb]{\smash{{\SetFigFont{8}{9.6}{\rmdefault}{\mddefault}{\updefault}{\color[rgb]{0,0,0}$0$}}}}}
\put(-899,1064){\makebox(0,0)[lb]{\smash{{\SetFigFont{8}{9.6}{\rmdefault}{\mddefault}{\updefault}{\color[rgb]{0,0,0}$\ket{\pm i}$ or $\ket \pm$ input}}}}}
\end{picture} \end{center}

The qubit which must be transformed is denoted $\rho$ (upper left qubit) and the lower left qubit's state is chosen depending on whether the upper qubit should be turned into a dummy or not. We refer to~\cite{KKMO21:delegating} for details. In order to be correct, it requires an additional correction step after this computation. The correction depends on the state of the second input qubit and the measurement outcome of the last qubit in the lower line.

\begin{center}
\begin{tabular}{|c|c|c|c|}
  2nd qubit input & Outcome & Correction & Effect \\
  \hline
  $\ket{+_i}$ & 0 & $\Y$ & $\I$  \\ 
  $\ket{+_i}$ & 1 & $\Y$ & $\I$  \\
  $\ket{-_i}$ & 0 & $\I$ & $\I$  \\ 
  $\ket{-_i}$ & 1 & $\I$ & $\I$  \\
  $\ket + $   & 0 & $\X$ & $\Ha$ \\ 
  $\ket + $   & 1 & $\Z$ & $\Ha$ \\
  $\ket - $   & 0 & $\Z$ & $\Ha$ \\ 
  $\ket - $   & 1 & $\X$ & $\Ha$ \\
  \hline 
\end{tabular}
\end{center}

Unfortunately, this correction depends on the final state. More precisely, by flipping the value reported as measurement outcome, the Server can apply an $\Y$ operation on dummy qubits and leave the rest unaffected. This flips selectively the state of dummies only, even if the server does not know that the qubit being prepared is in fact a dummy. We show in the next section how this breaks the verifiability of the protocol. The main take-away is that any correction applied after the computation does not also depend on the final state.

A potential patch was constructed using a more compact setup which seemingly satisfies all conditions. The graph that is used consists of a three vertex line for each final qubit. The first qubit in the line is measured either with an angle $\pi/2$ for dummy vertices or $0$ for other positions. The second vertex will always be measured with an angle of $\pi/2$. This is presented below in Figure \ref{fig:h-i-gadget}.

\begin{figure}[ht]\centering
     \subfloat[Measurement pattern for rotated qubits.]{{\includegraphics[width=0.4\columnwidth]{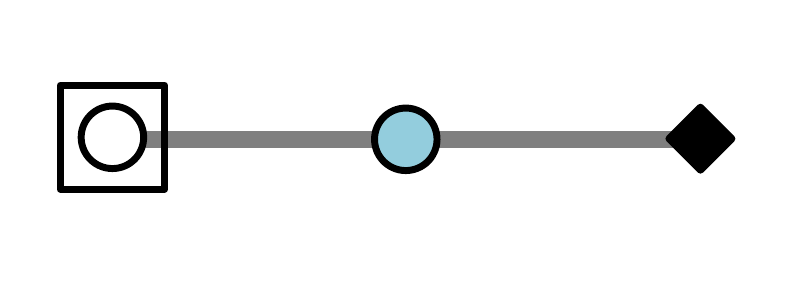} }}\qquad
     \subfloat[Measurement pattern for dummy qubits.]{{\includegraphics[width=0.4\columnwidth]{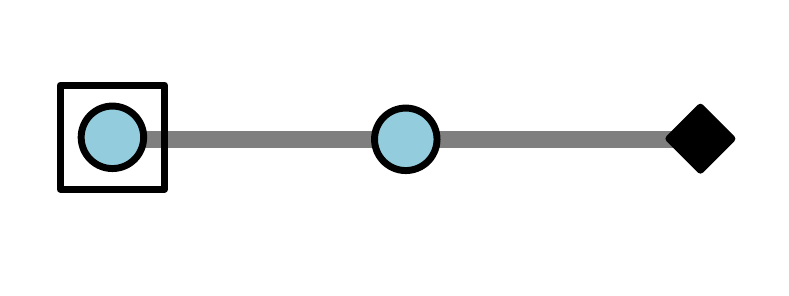} }}\caption{DBQC measurement pattern applied to each qubit in the verifiable graph. The vertices surrounded with squares are inputs, round vertices are measured, diamond vertices are outputs. Blue vertices correspond to a measurement angle of $\pi/2$ while white vertices are measured with angle $0$.}
     \label{fig:h-i-gadget}\end{figure}

In the first case, the operation that is applied is $\Ha\Z(\pi/2)\Ha\Z(\pi/2) = \X(\pi/2)\Z(\pi/2)$, which has the effect of transforming the state $\ket{+}$ first into $\ket{+_{\pi/2}}$ with the $Z$-rotation and then into $\ket{0}$ via the $X$-rotation. Note that this does not correspond to a Hadamard gate since it transforms $\ket{-}$ into $-i\ket{1}$, but it is sufficient for our purposes. In the second case, the operation correspond to $\X(\pi/2)$, which has no effect when applied to a $\ket{+}$ state up to a global phase. Since all qubits are rotated $\ket{+}$ states, the rotation of the last qubit in the three vertex line graph re-encodes the state if the final state is not a dummy. This yields the full set of states required by the SDQC protocol.

Note that once again, the conditions appear to be satisfied. Also, there are no post-processing steps beyond the standard MBQC flow corrections. However, here the input states do not span the full $8$ rotated states, but are always considered as $\ket{+}$ states and they are re-encrypted via the rotation of the last qubit. By applying $\Z$ on the input before the computation and the output after the computation, it is also possible to selectively flip dummies only: the two $\Z$s will cancel out for rotated qubits, but the second $\Z$ will have no effect on dummies while the first $\Z$ will flip the dummy.

\paragraph{Attack from Selective Dummy Flipping.}

We describe here an attack on the VBQC scheme of~\cite{KW17:optimised}, assuming that the Adversary can flip the value of dummy qubits (without affecting the computation and traps). We assume for this section knowledge about the Dotted-Triple Graph construction of~\cite{KW17:optimised}. Consider a line graph of two qubits and its transformation in a Dotted-Triple Graph. This graph contains two primary locations with three qubits and one added location with nine qubits. 

Through the application of $\CZ$ gates to construct the graph, flipping the value of a dummy is equivalent to applying $\Z$ on all adjacent qubits. For a given qubit in the graph, the global effect is $\Id$ if an even number of adjacent dummies are flipped, and $\Z$ if an odd number of adjacent dummies are flipped. As we wish to disrupt the computation but not affect traps, the key to our attack is to use the difference in the number of dummies in the neighbourhood of traps and computation qubits. Traps are only linked to dummies while a computation qubit will always have at least one other computation qubit among its neighbours. As shown in Figure \ref{fig:vbqc-att} below, we selectively flip added vertices so that each primary vertex is linked to exactly two attacked added vertices.

\begin{figure}[ht]
    \centering
	\includegraphics[width=0.8\columnwidth]{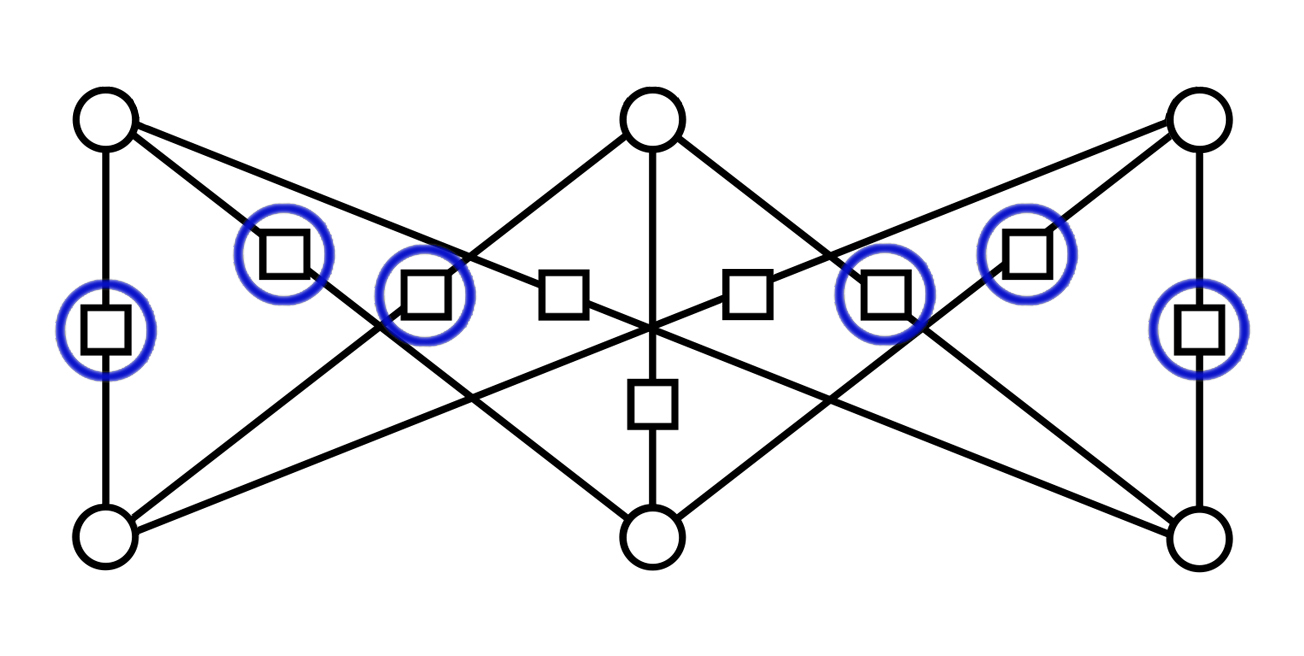}
    \caption{Example of attack layout where each top and bottom primary qubit is attached to exactly two attacked added qubits. Qubits that have been chosen for the attack are circled in blue.}
    \label{fig:vbqc-att}
\end{figure}

In that case, since the primary trap qubits are only linked to dummies, the attack does not trigger either trap (if one of the middle qubits that is attacked is a trap, the effect of the attack on this trap is $\Id$ as explained above). However, the attack may either affect two dummies linked to the primary computation qubits, in which case there is no attack since the effects cancel out, or one added dummy and the added computation qubit. Then, the effect on the added computation qubit is $\Id$ but the attacked dummies will apply a $\Z$ operation on primary computation qubits on both sides of the link. If we assume fixed (but unknown) attack positions, whether this attack succeeds in modifying the computation depends only on the colouring that is used, while never triggering any trap. The probability of success is equal to $2/3$: the attack succeeds if the computational added qubit is chosen for the attack, there are $6$ possible choices of attack configuration and each added qubit is left untouched by $2$ out of the $6$ attack configurations. We give in Figure \ref{fig:att-col} two possible colourings, ones in which the attack has no effect one the computation while the other corrupts it.

\begin{figure}[ht]\centering
    \subfloat[$\Z$ attack on both primary computational qubits due to odd number of attacked added dummies.]{{\includegraphics[width=0.45\columnwidth]{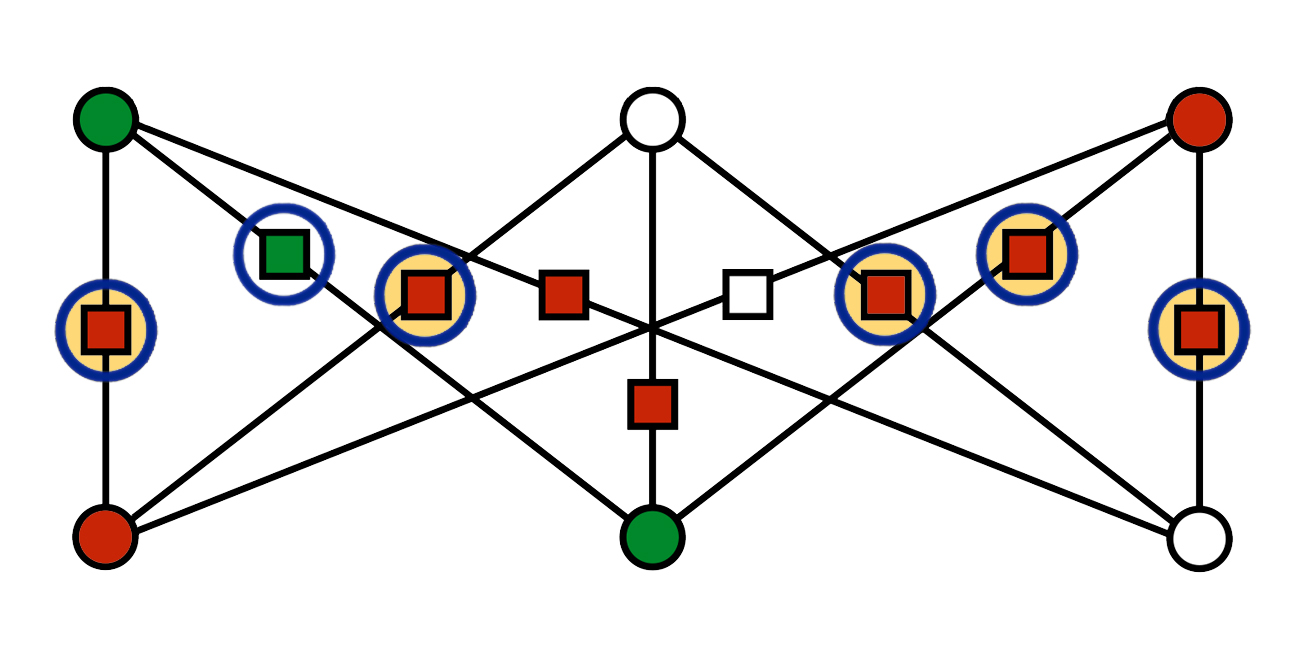} }}\qquad
    \subfloat[No attack on either primary computational qubit due to even number of attacked added dummies.]{{\includegraphics[width=0.45\columnwidth]{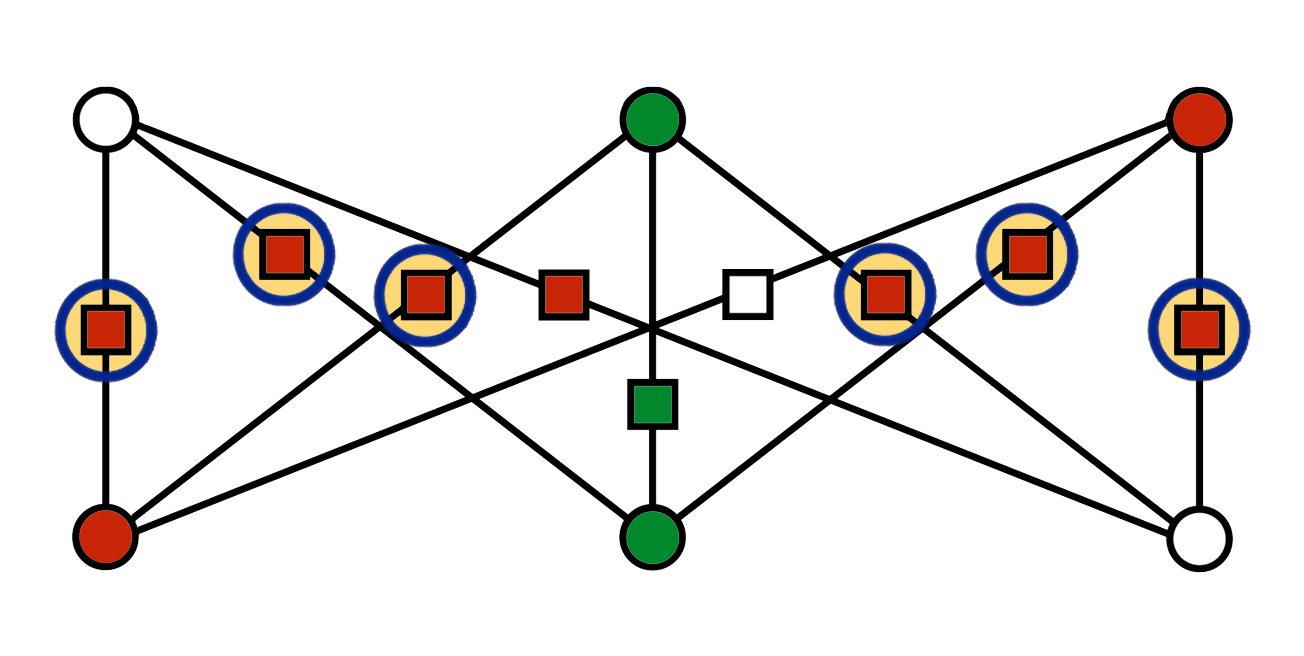} }}\caption{Two colourings of the previous graph (computational qubits are green, traps are white and dummies are red) for the same attacked qubits but a different effect on primary computational qubits. Attacked qubits are circled in blue, which translates to an $\X$ effect on dummies (yellow-filled circle) and no effect on added computational qubits (empty circle). The primary trap qubits are never affected by the attack since they are always attached to an even number of attacked added dummies.}\label{fig:att-col}\end{figure}

\paragraph{Extension and Take-away.}

Essentially, allowing an attack to depend on the nature of the qubits introduces new attacks compared to those that are possible in the plain VBQC Protocol. We have shown above that even a simple attack of this type is sufficient to break verifiability.

This attack is not specific to the construction of~\cite{KW17:optimised} and can also be applied to the robust SDQC Protocol of~\cite{LMKO21:verifying}. There, a trap is also always linked to dummies but computation qubits are never linked to dummies (the test and computation graphs are separated as in the current work). Applying the selective flip and apply a $\Z$ on all neighbours will corrupt the computation but leave traps unaffected.

There are two main ingredients to these attacks: (i) the possibility for the server to selectively attack qubits depending on whether they are dummies or rotated qubits, and (ii) the difference in the neighbourhoods of computation and trap qubits in terms of dummies. Regarding the first point, the sufficient conditions above are very restrictive as to what types of computations can be performed to generate a wider range of states. Together with the result from Section \ref{sec:impossibility}, this is a strong indication that starting from a Collaborative RSP for a restricted set of states and expanding it is hard to construct securely. As for the second point, if it is possible to construct an SDQC protocol in which the effect of flipping any number of dummies is the same on the tests and computations, then this attack would have no effect beyond what the SDQC protocol already protects. The current paper provides a solution to this problem by removing dummies altogether in the case of classical inputs and outputs. Other directions can be explored as well in order to construct a protocol which resits these attacks and handles quantum inputs and outputs.

We note that at no point do we break the theorem from \cite{KKMO21:delegating} proving the sufficient conditions for constructing an MBQC gadget for blindly generating an SDQC resource state up to state-independent deviations. However, it shows that the following conditions were implicitly assumed in the proof: (i) the starting states should span the full range of rotated $\ket{+_\theta}$ states, and (ii) any post-processing should be independent of the final state.

\end{document}